\documentclass[12pt,DIV=14]{scrartcl}
\usepackage{amsthm,amssymb,amsmath}
\usepackage{enumerate}
\usepackage
[colorlinks=true,citecolor=red,urlcolor=purple,linkcolor=blue,pdfborder={0 0 0}]%
{hyperref} % enables hyperlinks
\usepackage{tikz}

\allowdisplaybreaks[1]

\newtheorem{definition}{Definition}
\newtheorem{lemma}[definition]{Lemma}
\newtheorem{theorem}[definition]{Theorem}
\newtheorem{remark}[definition]{Remark}
\newtheorem{example}[definition]{Example}
\newtheorem{corollary}[definition]{Corollary}

\providecommand{\eps}{\varepsilon}
\providecommand{\C}{\mathbb{C}}

\providecommand{\N}{\mathbb{N}}

\providecommand{\R}{\mathbb{R}}
\providecommand{\Z}{\mathbb{Z}}

\providecommand{\id}{\operatorname{id}}
\providecommand{\inn}{\operatorname{int}}
\providecommand{\spann}{\operatorname{span}}
\providecommand{\supp}{\operatorname{supp}}
\let\Re\relax\DeclareMathOperator{\Re}{Re}
\let\Im\relax\DeclareMathOperator{\Im}{Im}

\providecommand{\Om}{\Omega}
\providecommand{\pOm}{{\partial\Omega}}

\providecommand{\vecr}{\hat\vecx}
\providecommand{\vecx}{\boldsymbol{x}}

\providecommand{\vecnu}{\boldsymbol{\nu}}

\renewenvironment{thebibliography}[1]{
\begin{oldthebibliography}{#1}
\setlength{\itemsep}{0em}
\setlength{\parskip}{0em}
}
{
\end{oldthebibliography}
}
\title{A radiation and propagation problem for a Helmholtz equation
with a compactly supported nonlinearity}
\author{Lutz Angermann\thanks{%
Dept.~of Mathematics, Clausthal University of Technology, Erzstr.~1, D-38678 Clausthal-Zellerfeld, Germany,
lutz.angermann@tu-clausthal.de}}
\date{\today}
\begin{document}

\maketitle

\begin{abstract}
%In memory of Vasyl V. Yatsyk
The present work describes some extensions of an approach,
originally developed by V.V. Yatsyk and the author, for the theoretical and numerical 
analysis of scattering and radiation effects on infinite plates with cubically polarized layers.
The new aspects lie on the transition to more generally shaped,
two- or three-dimensional objects,
which no longer necessarily have to be represented in terms a Cartesian product of real intervals,
to more general nonlinearities (including saturation) and
the possibility of an efficient numerical approximation of the electromagnetic fields
and derived quantities (such as energy, transmission coefficient, etc.). 
The paper advocates an approach that consists in transforming the original full-space problem
for a nonlinear Helmholtz equation (as the simplest model)
%for a system of nonlinear partial differential equations
into an equivalent boundary-value problem on a bounded domain by means
of a nonlocal Dirichlet-to-Neumann (DtN) operator.
It is shown that the transformed problem is equivalent to the original one
and can be solved uniquely under suitable conditions.
Moreover, the impact of the truncation of the DtN operator
on the resulting solution is investigated,
so that the way to the numerical solution by appropriate finite element methods is available. 
\end{abstract}

%\noindent\bigskip
\textsf{Keywords:}
Scattering, radiation,
nonlinear Helmholtz equation,
nonlinearly polarizable medium,
DtN operator,
truncation

%\noindent
\bigskip
\textsf{AMS Subject Classification (2022):}
35\,J\,05 % Laplace operator, Helmholtz equation (reduced wave equation), Poisson equation 
35\,Q\,60 % PDEs in connection with optics and electromagnetic theory
78\,A\,45 % Diffraction, scattering

\section{Introduction}
The present work deals with the mathematical modeling of the response of
a penetrable two- or three-dimensional object (obstacle), represented by a bounded domain,
to the excitation  by an external electromagnetic field.
A special aspect of the paper is that, in contrast to many other, thematically comparable works,
nonlinear constitutive laws of this object are in the foreground.

A standard example are the so-called Kerr nonlinearities.
It is physically known, but also only little investigated mathematically
that sufficiently strong incident fields, under certain conditions, cause effects
such as frequency multiplication, which cannot occur in the linear models
frequently considered in the literature.
On the other hand, such effects are interesting in applications,
which is why a targeted exploitation, for example from a numerical or
optimization point of view, first requires thorough theoretical investigation.

A relatively simple mathematical model for this is a nonlinear Helmholtz equation,
which results from the transition from the time-space formulation of Maxwell's equations
to the frequency-space formulation together with further simplifications.
Although some interesting nonlinear effects cannot be modeled by means of a single scalar equation
alone, its investigation is of own importance, for example from the aspect of variable coefficients,
and on the other hand its understanding is also the basis for further development,
for example for systems of nonlinear Helmholtz equations, see, e.g., \cite{Angermann:19a}.
The latter is also the reason why we consider a splitted nonlinearity
and not concentrate the nonlinearity in one term as is obvious.

The Helmholtz equation with nonlinearities has only recently become the focus
of mathematical investigations.
However, problems are mainly dealt with in which the nonlinearities are globally smooth,
while here a formulation as a transmission problem is used that allows
less smooth transitions at the object boundary.
In addition, we allow more general nonlinearities than the Kerr nonlinearities mentioned,
in particular saturation effects can be taken into account.

Starting from a physically oriented problem description as a full-space problem,
we derive a weak formulation on a bounded domain using the well-known technique
of DtN operators, and show its equivalence to the weakly formulated original problem.
Since the influence of the external field only occurs indirectly
in the weak formulation, we also give a second variant of the weak formulation
that better clarifies this influence and which we call the input-output formulation.

Since the DtN operators are non-local, their practical application (numerics) causes problems,
which is why a well-known truncation technique is used.
This raises the problem of proving the well-posedness of the reduced problem
and establishing a connection (error estimate) of the solution of the reduced problem
to the original problem.
Although these questions in the linear case have been discussed in the literature
for a relatively long time, they even for the linear case seemed to have been treated
only selectively and sometimes only very vaguely.
The latter concerns in particular the question of the independence of the stability constant
from the truncation parameter. In this work, both stability and error estimates are given
for the two- and three-dimensional case, whereby a formula-based relationship
between the discrete and the continuous stability constant is established.

Another difference to many existing, especially older works is that the present paper works
with variational (weak) formulations but not with integral equations. Unfortunately,
the complete tracking of the dependence of the occurring parameters on the wave number
(so-called wavenumber-independent bounds) has not yet been included.

It has already been mentioned that, for the linear situation,
in connection with scattering problems or with problems that are formulated
from the very beginning in bounded domains (e.g., with impedance boundary conditions),
there is an extensive and multi-threaded body of literature that is beyond
the scope of this article to list.
Transmission problems of the type considered here are rarely found in the literature.

Nevertheless, without claiming completeness, a few works should be mentioned here that
had an influence on the present results and whose bibliographies may be of help.
A frequently cited work that deals with linear scattering problems in two dimensions
and also served as the motivation for the present work is \cite{Hsiao:11},
which, however, does not discuss the dependence of the stability constant on the truncation parameter.
A number of later works by other authors quote this work, but sometimes assume results
that cannot be found in the original.
In this context, the papers \cite{Li:20b} and \cite{Xu:21} should also be mentioned,
which take up and improve various aspects of \cite{Hsiao:11}, e.g.,
the convergence order (exponential convergence of the truncated solution to the origibnal one).
However, they are also restricted to linear two-dimensional scattering or transmission problems,
respectively.

It is also worth noting that, in addition to the DtN-type methods,
there are other methods for reducing full-space problems to problems in bounded domains, too.
These include, above all, the so-called PML methods. Among the works related to the present work,
\cite{Jiang:22} should be mentioned, which considers Kerr nonlinearities and,
using a linearization approach, can show exponential convergence of the PML problem
in relation to the PML parameters.

The work that comes closest to our intentions is \cite{Koyama:07},
where the exterior Dirichlet boundary-value problem for the linear Helmholtz equation
is considered.
In this paper, no separate, parameter-uniform stability estimate of the truncated problem
is given, but the truncation error is included in the error estimate
of a finite element approximation.
A similar work is \cite{Koyama:09}, but in which another boundary condition at the boundary
of the auxiliary domain is considered, the so-called modified DtN condition.

Among the more recent papers, works by Mandel \cite{Mandel:19},
Chen, Ev{\'{e}}quoz \& Weth \cite{Chen:21}, and Maier \& Verf\"{u}rth \cite{Maier:22}
should be mentioned, especially because of the cited sources.
In his cumulative habilitation thesis, which contains further references,
Mandel examines existence and uniqueness questions for solutions of systems
of nonlinear Helmholtz equations in the full-space case.
Scattering or transmission problems are not considered.
Using integral operators, Chen et al.\ study the scattering problem with
comparatively high regularity assumptions to the superlinear nonlinearities
by means of topological fixed point and global bifurcation theory
to prove the existence of bounded solutions,
avoiding truncation approaches.
In this context also the paper \cite{Evequoz:14} is worth to be mentioned, in which
the existence of real-valued solutions (which satisfy certain asymptotic conditions)
of a nonlinear Helmholtz equation with a compactly supported nonlinearity is investigated.

Maier \& Verf\"{u}rth, who focus mainly on multiscale aspects for a nonlinear Helmholtz equation
over a bounded domain with impedance boundary conditions,
give an instructive review of the literature on nonlinear Helmholtz equations.
Further works on nonlinear Helmholtz equations in bounded domains are \cite{Wu:18}
and \cite{Verfuerth:23}, which deal with Kerr nonlinearities and use explicit iterative arguments.

A number of papers on inverse problems also deal with nonlinear Helmholtz equations,
even if the questions answered there are not very closely related to ours.
Typically, such treatises also include statements about the direct problem,
and so we mention here \cite{Griesmaier:22}, \cite{Jalade:04}, and \cite{Harrach:23},
where the latter work considers $C^\infty$-bounded domains and real solutions.

The structure of the present work is based on the program outlined above.
After the problem formulation in Section 2, the exterior auxiliary problem
required for truncation is discussed, after which the weak formulation and
equivalence statement follow in Section 4. Section 5 is dedicated to
the existence and uniqueness of the weak solution, where in particular
the assumptions on the nonlinear terms are discussed.
The final section then deals with the properties of the truncated problem
-- uniform (with respect to the truncation parameter) well-posedness
and estimate of the truncation error.

\section{Problem formulation}
\label{sec:problem}
Let $\Om\subset\R^d$ be a bounded domain with a Lipschitz boundary $\pOm$.
It represents a medium with a nonlinear behaviour with respect to electromagnetic fields.
Since $\Om$ is bounded, we can choose an open Euclidean $d$-ball $B_R\subset\R^d$ of radius
$R > \sup_{\vecx\in\Om}|\vecx|$
with center in the origin such that $\Om\subset B_R$.
The complements of $\Om$ and $B_R$ are denoted by
$\Om^\mathrm{c} := \R^d\setminus\Om$
$B_R^\mathrm{c} := \R^d\setminus B_R$, resp.,
the open complement of $B_R$ is denoted by $B_R^+ := \R^d\setminus\overline{B_R}$
(the overbar over sets denotes their closure in $\R^d$),
and the boundary of $B_R$, the sphere, by $S_R := \partial B_R$ (cf.\ Fig.~\ref{fig:img01}).
The open complement of $\Om$ is denoted by $\Om^+ := \R^d\setminus\overline{\Om}$.
By $\vecnu$ we denote the outward-pointing (w.r.t.\ either $\Om$ or $B_R$) unit normal vector
on $\pOm$ or $S_R$, respectively.

\begin{figure}[htb]
\centering{%
\begin{tikzpicture}
\draw [gray,fill,opacity=.4] plot [smooth cycle] coordinates {(0,0) (1,1) (3,1) (2,-1)};
\draw[dashed] (1.5,0) circle (2);
\draw[->] (-.8,2.1) -- (1.2,1.1);
\node at (2,0) {$\Om$};
\node at (3.8,-1) {$S_R$};
\node at (-.7,1.6) {$u^\mathrm{inc}$};
\end{tikzpicture}}
\caption{The nonlinear medium $\Om$ is excited by an incident field $u^\mathrm{inc}$ ($d=2$)}
\label{fig:img01}
\end{figure}
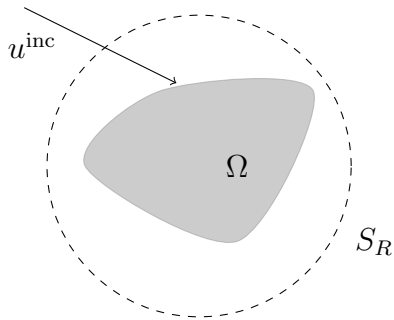

Trace operators will be denoted by one and the same symbol $\gamma$; the concrete
meaning (e.g., traces on the common interface of an interior and exterior domain)
will be clear from the context.

With regard to the function spaces used, we refer to the relevant literature,
e.g., \cite{Adams:03},
\cite[Ch.~1]{Lions:72a},
\cite[Ch.~3]{McLean:00}.
The corresponding norms are denoted by $\|\cdot\|_{0,p,\Om}$ for the $L_p(\Om)$-norm
and $\|\cdot\|_{s,2,\Om}$ for the $H^s(\Om)$-norm (as representative examples;
other domains of definition may also occur).

The classical direct problem of radiation and propagation of an electromagnetic field
-- actually just one component of it --
by/in the penetrable obstacle $\Om$ is governed by a nonlinear Helmholtz equation
with a variable complex-valued wave coefficient:
\begin{equation}\label{eq:genscalproblem}
-\Delta u(\vecx) - \kappa^2 c(\vecx,u) \,u = f(\vecx,u)
\quad\text{for (almost) all }
\vecx\in\R^d,
\end{equation}
where the wavenumber $\kappa>0$ is fixed.
The physical properties of the obstacle $\Om$ are described by the coefficient
$c:\;\R^d\times\C\to\C$ (physically the square of the \emph{refractive index}) and
the right-hand side $f:\;\R^d\times\C\to\C$.
In general, both functions are nonlinear and have the following properties:
\begin{equation}\label{eq:ass_nonlin1}
\supp(1-c(\cdot,w)) = \overline\Om
\quad\text{and}\quad
\supp f(\cdot,w) \subset\overline\Om
\quad\text{for all }
w\in\C.
\end{equation}
The function $1-c$ is often called the \emph{contrast function}.
Basically we assume that $c$ and $f$ are Carath\'{e}odory functions,
i.e.\
the mapping $\vecx\mapsto c(\vecx,v)$ is (Lebesgue-)measurable for all $v\in\C$,
and
the mapping $v\mapsto c(\vecx,v)$ is continuous for almost all $\vecx\in\R^d$.
These two conditions imply that $\vecx\mapsto c(\vecx,v(\vecx))$
is measurable for any measurable $v$.
The same applies to $f$.

The unknown \emph{total field} $u:\;\R^d\to\C$ should have the following structure:
\begin{equation}\label{eq:solstructure}
u=\begin{cases}
u^\mathrm{rad} + u^\mathrm{inc} & \text{in }\Om^\mathrm{c},\\
u^\mathrm{trans} & \text{in }\Om,
\end{cases}
\end{equation}
where
$u^\mathrm{rad}:\;\Om^\mathrm{c}\to\C$ is the unknown radiated/scattered field,
$u^\mathrm{trans}:\;\Om\to\C$ denotes the unknown transmitted field,
and the incident field $u^\mathrm{inc}\in H^1_\mathrm{loc}(\Om^+)$ is given.
The incident field is usually a (weak) solution of either the homogeneous or
inhomogeneous Helmholtz equation (even in the whole space).
Typically it is generated either by concentrated sources located
in a bounded region of $\Om^+$ or by sources at infinity, e.g.\ travelling waves.

\begin{example}[$d=2$]
The incident plane wave, whose transmission and scattering is investigated,
is given by
\[
u^\mathrm{inc}(\vecx) := \alpha^\mathrm{inc} \exp (i(\Phi x_1 - \Gamma x_2)),
\vecx = (x_1,x_2)^\top\in B_R^+
\]
with amplitude $\alpha^\mathrm{inc}$ and angle of incidence
$\varphi^\mathrm{inc}$, $|\varphi^\mathrm{inc}| < \pi$, where
$\Phi := \kappa\sin\varphi^\mathrm{inc}$ is the longitudinal
wave number and $\Gamma := \sqrt{\kappa^2 - \Phi^2} = \kappa\cos\varphi^\mathrm{inc}$
the transverse wave number.
In polar coordinates is then
\begin{align*}
u^\mathrm{inc}(r,\varphi)
&= \alpha^\mathrm{inc} \exp (i(\Phi r \cos\varphi - \Gamma r \sin\varphi))\\
&= \alpha^\mathrm{inc} \exp (i\kappa r(\sin\varphi^\mathrm{inc} \cos\varphi
- \cos\varphi^\mathrm{inc} \sin\varphi))\\
&= \alpha^\mathrm{inc} \exp (i\kappa r \sin(\varphi^\mathrm{inc} - \varphi)),
\quad (r,\varphi) \in B_R^+.
\end{align*}
\end{example}

The radiated/scattered field $u^\mathrm{rad}$
should satisfy an additional condition, the so-called \emph{Sommerfeld radiation condition}:
\begin{equation}\label{eq:sommerfeld1}
\lim_{|\vecx|\to\infty} |\vecx|^{(d-1)/2}
\left(\vecr\cdot\nabla u^\mathrm{rad} - i\kappa u^\mathrm{rad}\right) = 0
\end{equation}
uniformly for all directions $\vecr:=\vecx/|\vecx|$, where
$\vecr\cdot\nabla u^\mathrm{rad}$ denotes the derivative of $u^\mathrm{rad}$
in radial direction $\vecr$,
cf.\ \cite[eq.~(3.7) for $d=3$, eq.~(3.96) for $d=2$]{Colton:13b}.
Physically, the condition \eqref{eq:sommerfeld1} allows only \emph{outgoing} waves at infinity;
mathematically it guarantees the uniqueness of the solution $u^\mathrm{scat}:\;B_R^+ \to\C$
of the following exterior Dirichlet problem
\begin{equation}\label{eq:extDirproblem}
\begin{aligned}
-\Delta u^\mathrm{scat} - \kappa^2 u^\mathrm{scat} = 0
\quad\text{in }
B_R^+,\\
u^\mathrm{scat} = f_{S_R}
\quad\text{on }
S_R,\\
\lim_{|\vecx|\to\infty} |\vecx|^{(d-1)/2}
\left(\vecr\cdot\nabla u^\mathrm{scat} - i\kappa u^\mathrm{scat}\right) = 0,
\end{aligned}
\end{equation}
where $f_{S_R}:\;S_R\to\C$ is given.
We mention that, in the context of classical solutions (i.e.\ $u^\mathrm{scat}\in C^2(B_R^+)$) to
problem \eqref{eq:extDirproblem},
Rellich \cite{Rellich:43} has shown that the condition \eqref{eq:sommerfeld1}
can be weakened to the following integral version:
\[
\lim_{|\vecx|\to\infty} \int_{S_R}
\left|\vecr\cdot\nabla u^\mathrm{scat} - i\kappa u^\mathrm{scat}\right|^2 ds(\vecx) =0.
\]
In the context of weak solutions (i.e.\ $u^\mathrm{scat}\in H^1_\mathrm{loc}(B_R^+)$),
an analogous equivalence statement can be found in \cite[Thm.~9.6]{McLean:00}.
\section{The exterior problem in $B_R^\mathrm{c}$}
For a given $f_{S_R}\in C(S_R)$ and $d=3$,
the unique solvability of problem \eqref{eq:extDirproblem} in $C^2(B_R^+)\cap C(B_R^\mathrm{c})$
is proved, for example, in \cite[Thm.~3.21]{Colton:13b}.
In addition, if $f_{S_R}$ is smoother, say $f_{S_R}\in C^\infty(S_R)$, then
the normal derivative of $u^\mathrm{scat}$ on the boundary $S_R$ is a well-defined
continuous function \cite[Thm.~3.27]{Colton:13b}.
These assertions remain valid in the case $d=2$, see \cite[Sect.~3.10]{Colton:13b}.

Therefore, by solving \eqref{eq:extDirproblem} for given $f_{S_R}\in C^\infty(S_R)$,
a mapping can be introduced that takes the Dirichlet data on $S_R$
to the corresponding Neumann data on $S_R$, i.e.
\[
f_{S_R}\mapsto T_\kappa f_{S_R} := \left.\vecr\cdot\nabla u^\mathrm{scat}\right|_{S_R},
\]
see, e.g., \cite[Sect.~3.2]{Colton:19}.

Furthermore, it is well-known that the mapping $T_\kappa$ can be extended to a bounded linear operator
$T_\kappa:\; H^{s+1/2}(S_R)\to H^{s-1/2}(S_R)$ for any $|s|\le 1/2$ \cite[Thm.~2.31]{ChandlerWilde:12}
(we keep the notation already introduced for this continued operator).
This operator is called the \emph{Dirichlet-to-Neumann operator}, in short \emph{DtN operator},
or \emph{capacity operator}.

Since the problem \eqref{eq:extDirproblem} is considered in a spherical exterior domain,
an explicit series representation of the solution is available
using standard separation techniques in polar or spherical coordinates, respectively.
The term-by-term differentiation of this series thus also provides a series representation
of the image of $T_\kappa$.
We first give a formal description of the approach and then comment on its
mathematical correctness at the end.

The solution of the problem \eqref{eq:extDirproblem} in the two-dimensional case
(here with $u^\mathrm{scat}$ replaced by $u$)
is given by \cite[Proposition~2.1]{Masmoudi:87}, \cite[eq.~(30)]{Keller:89}:
\begin{equation}\label{eq:Fourier2d}
\begin{aligned}
u(\vecx)=u(r\hat\vecx)=u(r,\varphi)
= \sum_{n\in\Z} \frac{H_n^{(1)}(\kappa r)}{H_n^{(1)}(\kappa R)}\, f_n(R) Y_n(\hat\vecx)
&= \sum_{n\in\Z} \frac{H_n^{(1)}(\kappa r)}{H_n^{(1)}(\kappa R)}\, f_n(R) Y_n(\varphi),\\
&\quad
\vecx=r\hat\vecx\in S_r,\ r>R,\ \varphi\in[0,2\pi]
\end{aligned}
\end{equation}
(identifying $u(\vecx)$ with $u(r,\varphi)$ and $Y_n(\hat\vecx)$ with $Y_n(\varphi)$ for
$\vecx=r\hat\vecx=r(\cos\varphi,\sin\varphi)^\top$),
where $(r,\varphi)$ are the polar coordinates,
$H_n^{(1)}$ are the cylindrical Hankel functions of the first kind of order $n$
\cite[Sect.~10.2]{NIST:22}\footnote{%
Instead of \eqref{eq:sommerfeld1} \cite{Masmoudi:87} considered the ingoing Sommerfeld condition
and thus obtained a representation in terms of the cylindrical Hankel functions of the second kind.
Note that
$H_n^{(2)}(-\xi)=-(-1)^n H_n^{(1)}(\xi)$
\cite[(10.11.5)]{NIST:22}.},
$Y_n$ are the circular harmonics defined by
\[
Y_n(\varphi) = \frac{e^{in\varphi}}{\sqrt{2\pi}},
\quad
n\in\Z,
\]
$f_n(R)$ are the Fourier coeﬃcients of $f_{S_R}$ defined by
\begin{equation}\label{eq:Fouriercoeff2d}
f_n(R) := (f_{S_R}(R\cdot),Y_n)_{S_1} = \int_{S_1}f_{S_R}(R\hat\vecx)\overline{Y_n}(\hat\vecx)ds(\hat\vecx)
=\int_0^{2\pi} f_{S_R}(R,\varphi)\overline{Y_n}(\varphi)d\varphi,
\end{equation}
and $ds(\hat\vecx)$ is the Lebesgue arc length element.
The terms of the series \eqref{eq:Fourier2d} are well-defined, since the main branch
of the Hankel functions $H_\nu^{(1)}$ of real order $\nu\ge 0$ is free of real zeros
\cite[p.~62]{Bateman:81b}.

Now we formally differentiate the representation \eqref{eq:Fourier2d} with respect to $r$
to obtain the outward normal derivative of $u$:
\[
\vecr\cdot\nabla u(\vecx) = \frac{\partial u}{\partial r}(r\hat\vecx)
= \kappa\sum_{n\in\Z} \frac{H_n^{(1)'}(\kappa r)}{H_n^{(1)}(\kappa R)}\, f_n(R) Y_n(\hat\vecx),
\quad
\vecx=r\hat\vecx\in S_r,\ r>R.
\]
Setting $f_R:=u|_{S_R}$ and letting $\vecx$ in this representation approach the boundary $S_R$,
we can formally define the (extended) DtN operator by
\begin{equation}\label{def:2dDtNb}
\begin{aligned}
T_\kappa u(\vecx)
&:= \frac{1}{R}\sum_{n\in\Z} Z_n(\kappa R)u_n(R)Y_n(\hat\vecx),
\quad
\vecx=R\hat\vecx\in S_R,
\end{aligned}
\end{equation}
where
\[
Z_n(\xi):= \xi\,\frac{H_n^{(1)'}(\xi)}{H_n^{(1)}(\xi)}\,,
\]
and $u_n(R)$ are the Fourier coefficients of $u|_{S_R}$
analogously to \eqref{eq:Fouriercoeff2d}.
The admissibility of this procedure has been proven in many sources in the classical context,
for example \cite[Sect.~3.5]{Colton:19}.
For the present case, in the paper \cite[Thm.~1]{Ernst:96} it was shown that
the operator $T_\kappa:\; H^{s+1/2}(S_R)\to H^{s-1/2}(S_R)$ is bounded for any $s\in\N_0$.
Ernst's result was extended to all $s\ge 0$ in \cite[Thm.~3.1]{Hsiao:11}.

In the case $d=3$, the solution of the problem \eqref{eq:extDirproblem}
is given by \cite[eq.~(33)]{Keller:89}:
\begin{equation}\label{eq:Fourier3d}
\begin{aligned}
u(\vecx)=u(r\hat\vecx)=u(r,\varphi,\theta)
&= \sum_{n\in\N_0}\sum_{|m|\le n}
\frac{h_n^{(1)}(\kappa r)}{h_n^{(1)}(\kappa R)}\,
f_n^m(R)Y_n^m(\hat\vecx)\\
&= \sum_{n\in\N_0}\sum_{|m|\le n}
\frac{h_n^{(1)}(\kappa r)}{h_n^{(1)}(\kappa R)}\,
f_n^m(R)Y_n^m(\varphi,\theta),\\
&\quad
\vecx\in S_r,\ r>R,\ (\varphi,\theta)\in [0,2\pi]\times [0,\pi]
\end{aligned}
\end{equation}
(identifying $u(\vecx)$ with $u(r,\varphi,\theta)$ and $Y_n^m(\hat\vecx)$
with $Y_n^m(\varphi,\theta)$ for
$\vecx=r\hat\vecx=r(\cos\varphi\sin\theta,\linebreak[0] \sin\varphi\sin\theta,\linebreak[0] \cos\theta)^\top$),
where $(r,\varphi,\theta)$ are the spherical coordinates,
$h_n^{(1)}$ are the spherical Hankel functions of the first kind of order $n$
\cite[Sect.~10.47]{NIST:22},
$Y_n^m$ are the spherical harmonics defined by
\[
Y_n^m(\varphi,\theta) = \sqrt{\frac{2n+1}{4\pi}\,\frac{(n-|m|)!}{(n+|m|)!}}\,
P_n^{|m|}(\cos\theta)e^{im\varphi},
\quad
n\in\N_0,\ |m|\le n,
\]
(identifying $Y_n^m(\hat\vecx)$ with $Y_n^m(\varphi,\theta)$ for
$\hat\vecx=(\cos\varphi\sin\theta,\sin\varphi\sin\theta,\cos\theta)^\top$),
where
$P_n^m$ are the associated Legendre functions of the first kind \cite[Sect.~14.21]{NIST:22},
$f_n^m(R)$ are the Fourier coeﬃcients defined by
\begin{equation}\label{eq:Fouriercoeff3d}
\begin{aligned}
f_n^m(R) = (f_{S_R}(R\cdot),Y_n^m)_{S_1}
&= \int_{S_1}f_{S_R}(R\hat\vecx)\overline{Y_n^m}(\hat\vecx)ds(\hat\vecx)\\
&=\int_0^{2\pi}\int_0^\pi f_{S_R}(R,\varphi,\theta)
\overline{Y_n^m}(\varphi,\theta)\sin\theta d\theta d\varphi,
\end{aligned}
\end{equation}
and $ds(\hat\vecx)$ is the Lebesgue surface area element.
The relationship
\[
h_n^{(1)}(x)=\sqrt{\frac{\pi}{2x}}H_{n+1/2}^{(1)},
\quad
n\in\N_0,
\]
(see, e.g., \cite[Sect.~10.47]{NIST:22})
and the above remark about the lack of real zeros of the Hankel functions $H_\nu^{(1)}$
of real order $\nu\ge 0$ imply that the terms of the series \eqref{eq:Fourier3d} are well-defined.

Proceeding as in the two-dimensional case, we get
\[
\vecr\cdot\nabla u(\vecx) = \frac{\partial u}{\partial r}(r\hat\vecx)
= \kappa\sum_{n\in\N_0}\sum_{|m|\le n}
\frac{h_n^{(1)}(\kappa r)}{h_n^{(1)}(\kappa R)}\,
f_n^m(R)Y_n^m(\hat\vecx),
\quad
\vecx=r\hat\vecx\in S_r,\ r>R.
\]

Setting $f_R:=u|_{S_R}$ and letting $r\to R$,
we can define the (extended) DtN operator by
\begin{equation}\label{def:3dDtNb}
\begin{aligned}
T_\kappa u(\vecx)
&=\frac{1}{R}\sum_{n\in\N_0}\sum_{|m|\le n}
z_n(\kappa R)u_n^m(R)Y_n^m(\hat\vecx),
\quad
\vecx=R\hat\vecx\in S_R,
\end{aligned}
\end{equation}
where
\[
z_n(\xi):= \xi\,\frac{h_n^{(1)'}(\xi)}{h_n^{(1)}(\xi)}\,,
\]
and $u_n^m(R)$ are the Fourier coefficients of $u|_{S_R}$
analogously to \eqref{eq:Fouriercoeff3d}.
The admissibility of this procedure is proved in \cite[Thm.~2.15]{Colton:19} or
\cite[Thm.~2.6.2]{Nedelec:01}, for example. For the present situation
there is a boundedness result for $d=3$ analogous to \cite[Thm.~3.1]{Hsiao:11}
in \cite[Thm.~2.6.4]{Nedelec:01}.
In summary, the following statement applies to both dimensions.
\begin{theorem}\label{th:DtNprop}
The DtN operator $T_\kappa:\; H^{s+1/2}(S_R)\to H^{s-1/2}(S_R)$ is bounded
for any $s\ge 0$.
\end{theorem}
\begin{remark}\label{rem:sharpDtNbound}
A more refined analysis of the DtN operator in the case $s=0$ results in a sharp estimate
of its norm w.r.t.\ the wavenumber \cite[Thm.~1.4]{Baskin:16}:
Given $\kappa_0>0$, there exists a constant $C>0$ independent of $\kappa$ such that
\[
\|T_\kappa v\|_{-1/2,2,S_R} \le C\kappa \|v\|_{1/2,2,S_R}
\quad\text{for all }
v\in H^1_\mathrm{loc}(B_R^+)
\quad\text{and}\quad
\kappa \ge \kappa_0.
\]
\end{remark}
The result from \cite[Thm.~1.4]{Baskin:16} applies to more general domains,
for the present situation it already follows from the proof of Lemma~\ref{l:Tkuniformcontinuity}
(see the estimates \eqref{eq:TkuniformcontinuitySR2}, \eqref{eq:TkuniformcontinuitySR3}
for $s=0$, where the bounds do not depend on $N$).

At the end of this section we give a collection of some properties of the
coefficient functions in the representations \eqref{def:2dDtNb}, \eqref{def:3dDtNb}
which will be used in some of the subsequent proofs.

\begin{lemma}\label{l:propHankelfunctions}
For all $\xi>0$, the following holds:
\begin{gather*}
-n \le \Re Z_n(\xi) \le -\frac12,
\quad
0 < \Im Z_n(\xi) < \xi
\quad\text{for all }
|n|\in\N,\\
-\frac12 \le \Re Z_0(\xi) < 0,
\quad
\xi < \Im Z_0(\xi),
\\
-(n+1)\le \Re z_n(\xi) \le -1,
\quad
0 < \Im z_n(\xi) \le \xi
\quad\text{for all }
n\in\N,\\
\Re z_0(\xi) = -1,
\quad
\Im z_0(\xi) = \xi.
\end{gather*}
\end{lemma}
\begin{proof}
For the case $d=2$,
the estimates can be found in \cite[eq.~(2.34)]{Shen:07}.
The other estimates can be found in \cite[Thm.~2.6.1]{Nedelec:01}, see also
\cite[eqs.~(2.22), (2.23)]{Shen:07}.
Although only $0\le \Im z_n(\xi)$ is specified in the formulation of the cited theorem,
the strict positivity follows from the positivity of the function $q_\ell$
in \cite[eq.~(2.6.34)]{Nedelec:01}, as has been mentioned in \cite{Melenk:10}.
\end{proof}
\begin{corollary}\label{cor:propHankelfunctions}
For all $\xi>0$, the following holds:
\begin{align*}
|Z_n(\xi)|^2 &\le (1+n^2)(1 + |\xi|^2)
\quad\text{for all }
|n|\in\N,\\
|z_n(\xi)|^2 &\le (1+n^2)(2 + |\xi|^2)
\quad\text{for all }
n\in\N_0.
\end{align*}
\end{corollary}
\begin{proof}
The estimates of the real and imaginary parts of $Z_n$ from Lemma~\ref{l:propHankelfunctions}
immediately imply that
\begin{align*}
\frac{1}{1+n^2} |Z_n(\xi)|^2
&= \frac{1}{1+n^2} \left[|\Re Z_n(\xi)|^2 + |\Im Z_n(\xi)|^2\right]\\
&\le \frac{1}{1+n^2} \left[n^2 + |\xi|^2\right]
\le 1 + \frac{|\xi|^2}{1+n^2}
\le 1 + |\xi|^2,
\quad
n\in\N.
\end{align*}
Since $H^{(1)}_{-n}(\xi)=(-1)^n H^{(1)}_{n}(\xi)$, $n\in\N$
\cite[eq.~(10.4.2)]{NIST:22}, the estimate is also valid for $n$ such that $-n\in\N$.

Analogously we obtain from Lemma~\ref{l:propHankelfunctions} that
\begin{align*}
\frac{1}{1+n^2} |z_n(\xi)|^2
&= \frac{1}{1+n^2} \left[|\Re z_n(\xi)|^2 + |\Im z_n(\xi)|^2\right]\\
&\le \frac{1}{1+n^2} \left[(1+n)^2 + |\xi|^2\right]
\le 2 + \frac{|\xi|^2}{1+n^2}
\le 2 + |\xi|^2.
\end{align*}
\end{proof}

\section{Weak formulations of the interior problem}

Now we turn to the consideration of the problem \eqref{eq:genscalproblem}--\eqref{eq:sommerfeld1}.

In the classical setting it can be formulated as follows:
Given $u^\mathrm{inc}\in H^1_\mathrm{loc}(\Om^+)$,
determine the transmitted field $u^\mathrm{trans}:\;\Om\to\C$ and
the radiated/scattered field $u^\mathrm{rad}:\;\Om^\mathrm{c}\to\C$
satisfying
\begin{equation}\label{eq:classicalintproblem}
\begin{aligned}
-\Delta u^\mathrm{trans} - \kappa^2 c(\cdot,u^\mathrm{trans}) \,u^\mathrm{trans} &= f(\cdot,u^\mathrm{trans})
&&\text{in }\Om,\\
-\Delta u^\mathrm{rad} - \kappa^2 u^\mathrm{rad} &= 0
&&\text{in }\Om^+,\\
u^\mathrm{trans} &= u^\mathrm{rad} + u^\mathrm{inc}
&&\text{on }\pOm,\\
\vecnu\cdot\nabla u^\mathrm{trans} &= \vecnu\cdot\nabla u^\mathrm{rad}
+ \vecnu\cdot\nabla u^\mathrm{inc}
&&\text{on }\pOm
\end{aligned}
\end{equation}
and the radiation condition \eqref{eq:sommerfeld1}.
Note that the incident field is usually a (weak) solution of either the homogeneous or
inhomogeneous Helmholtz equation in $\Om^+$, i.e.\ the second equation
in \eqref{eq:classicalintproblem} can be replaced by
\begin{equation}\label{eq:classicalintproblem2}
-\Delta u - \kappa^2 u = f^\mathrm{inc}\quad\text{in }\Om^+,
\end{equation}
where $f^\mathrm{inc}:\;\Om^+\to\C$ is an eventual source density.
For simplicity we do not include the case of a nontrivial source density
in our investigation, but the subsequent theory can be easily extended
by adding an appropriate linear functional, say $\ell^\mathrm{src}$,
on the right-hand side of the obtained weak formulations
(see \eqref{eq:weakfullspace} or \eqref{eq:weakball} later).

In order to give a weak formulation of \eqref{eq:classicalintproblem} with
the modification \eqref{eq:classicalintproblem2} in the case $f^\mathrm{inc}=0$,
we introduce the (complex) linear function spaces
\begin{align*}
H^1_\mathrm{comp}(\Om^+) &:= \left\{v\in H^1(\Om^+):
\;\supp v\ \text{ is compact}\right\},\\
V_{\R^d} &:= \{v\in L_2(\R^d):\; v|_\Om\in H^1(\Om),\ v|_{\Om^+}\in H^1_\mathrm{loc}(\Om^+):
\;\gamma v|_\Om = \gamma v|_{\Om^+} \text{ on }\pOm\},\\
V_{\R^d}^\circ &:= \{v\in L_2(\R^d):\; v|_\Om\in H^1(\Om),\ v|_{\Om^+}\in H^1_\mathrm{comp}(\Om^+):
\;\gamma v|_\Om = \gamma v|_{\Om^+} \text{ on }\pOm\}
\end{align*}
(note the comment at the beginning of Section \ref{sec:problem} on the notation for trace operators)
and multiply the first equation of \eqref{eq:classicalintproblem} by the restriction
$v|_\Om$ of an arbitrary element $v\in V_{\R^d}$ and \eqref{eq:classicalintproblem2}
by the restriction $v|_{\Om^+}$ of $v\in V_{\R^d}$, respectively, and integrate by parts:
\begin{align*}
(\nabla u^\mathrm{trans},\nabla v)_\Om - (\vecnu\cdot\nabla u^\mathrm{trans},\nabla v)_{\pOm}
- \kappa^2(c(\cdot,u^\mathrm{trans}) u^\mathrm{trans},v)_\Om
&=(f(\cdot,u^\mathrm{trans}),v)_\Om,\\
(\nabla u,\nabla v)_{\Om^+} - (\vecnu\cdot\nabla u,\nabla v)_{\partial\Om^+} - \kappa^2(u,v)_{\Om^+}
&=0.
\end{align*}
Here we use the notation, for any domain $M\subset\R^d$ with boundary $\partial M$
and appropriately defined functions on $M$ or $\partial M$,
\begin{align*}
(\nabla w,\nabla v)_M &:= \int_M \nabla w\cdot\nabla\overline{v} d\vecx,\\
(w,v)_M &:= \int_M w\overline{v} d\vecx,\\
(w,v)_{\partial M} &:= \int_{\partial M} w\overline{v} ds(\vecx)
\end{align*}
(the overbar over functions denotes complex conjugation).
Taking into consideration the last transmission condition in \eqref{eq:classicalintproblem},
the relationship $\vecnu|_\Om = - \vecnu|_{\Om^+}$, and the fact that the last but one
transmission condition in \eqref{eq:classicalintproblem} is included in the definition of
the space $V_{\R^d}$, we define a bivariate nonlinear form on
$V_{\R^d}\times V_{\R^d}^\circ$ by
\begin{align*}
a_{\R^d}(w,v)
&:=(\nabla w,\nabla v)_\Om + (\nabla w,\nabla v)_{\Om^+} - \kappa^2(c(\cdot,w) w,v)_{\R^d},
\end{align*}
cf., e.g., \cite[Example~21.8]{Wloka:87}.
\begin{definition}
Given $u^\mathrm{inc}\in H^1_\mathrm{loc}(\Om^+)$,
a \emph{weak solution} to the problem \eqref{eq:genscalproblem}--\eqref{eq:sommerfeld1}
is defined as an element $u\in V_{\R^d}$ that has the structure \eqref{eq:solstructure},
satisfies the variational equation
\begin{equation}\label{eq:weakfullspace}
a_{\R^d}(u,v)=(f(\cdot,u),v)_{\R^d}
\quad\text{for all }
v\in V_{\R^d}^\circ
\end{equation}
and the Sommerfeld radiation condition \eqref{eq:sommerfeld1}.
\end{definition}
A second weak formulation can be obtained if we do not replace
the second Helmholtz equation in \eqref{eq:classicalintproblem} by \eqref{eq:classicalintproblem2}.
Then the first step in the derivation of the weak formulation reads as
\begin{align*}
(\nabla u^\mathrm{trans},\nabla v)_\Om - (\vecnu\cdot\nabla u^\mathrm{trans},\nabla v)_{\pOm}
- \kappa^2(c(\cdot,u^\mathrm{trans}) u^\mathrm{trans},v)_\Om
&=(f(\cdot,u^\mathrm{trans}),v)_\Om,\\
(\nabla u^\mathrm{rad},\nabla v)_{\Om^+} - (\vecnu\cdot\nabla u^\mathrm{rad},\nabla v)_{\partial\Om^+}
- \kappa^2(u^\mathrm{rad},v)_{\Om^+}
&=0.
\end{align*}
The last transmission condition in \eqref{eq:classicalintproblem} allows to rewrite
the first equation as
\begin{gather*}
(\nabla u^\mathrm{trans},\nabla v)_\Om
- (\vecnu\cdot\nabla u^\mathrm{rad},\nabla v)_{\pOm}
- \kappa^2(c(\cdot,u^\mathrm{trans}) u^\mathrm{trans},v)_\Om\\
=(f(\cdot,u^\mathrm{trans}),v)_\Om + (\vecnu\cdot\nabla u^\mathrm{inc},\nabla v)_{\pOm},
\end{gather*}
leading to the weak formulation
\[
(\nabla u_0,\nabla v)_\Om + (\nabla u_0,\nabla v)_{\Om^+} - \kappa^2(c(\cdot,u_0) u_0,v)_{\R^d}
=(f(\cdot,u_0),v)_{\R^d}
+ (\vecnu\cdot\nabla u^\mathrm{inc},v)_\pOm
\quad\text{for all }
v\in V_{\R^d}^\circ
\]
with respect to the structure
\[
u_0:=\begin{cases}
u^\mathrm{rad} & \text{in }\Om^\mathrm{c},\\
u^\mathrm{trans} & \text{in }\Om,
\end{cases}
\]
where $u^\mathrm{rad}\in H^1_\mathrm{loc}(\Om^+)$, $u^\mathrm{trans}\in H^1(\Om)$.

The advantage of this formulation is that it clearly separates the unknown
and the known parts of the fields,
so we call this formulation the \emph{input-output formulation}.
The disadvantage is that the natural function space
of the solution $u_0$ is not a linear space due to the last but one
transmission condition in \eqref{eq:classicalintproblem}.

Instead of the problem \eqref{eq:genscalproblem}--\eqref{eq:sommerfeld1}
we want to solve an equivalent problem in the bounded domain $B_R$,
that is, we define
\[
V := \{v\in L_2(B_R):\; v|_\Om\in H^1(\Om),\ v|_{B_R\setminus\overline\Om}\in H^1(B_R\setminus\overline\Om):
\;\gamma v|_\Om = \gamma v|_{B_R\setminus\overline\Om} \text{ on }\pOm\}
\]
and look for an element $u\in V$ such that
\begin{equation}\label{eq:intscalproblem}
\begin{alignedat}{2}
-\Delta u^\mathrm{trans} - \kappa^2 c(\cdot,u^\mathrm{trans}) \,u &= f(\cdot,u^\mathrm{trans})
&&\quad\text{in }\Om,\\
-\Delta u - \kappa^2 u &= 0
&&\quad\text{in }B_R\setminus\overline\Om,\\
u^\mathrm{trans} &= u^\mathrm{rad} + u^\mathrm{inc}
&&\quad\text{on }\pOm,\\
\vecnu\cdot\nabla u^\mathrm{trans} &= \vecnu\cdot\nabla u^\mathrm{rad}
+ \vecnu\cdot\nabla u^\mathrm{inc}
&&\quad\text{on }\pOm,\\
\vecr\cdot\nabla u^\mathrm{rad} &= T_\kappa u^\mathrm{rad}
&&\quad\text{on } S_R
\end{alignedat}
\end{equation}
formally holds.
Now the weak formulation of problem \eqref{eq:intscalproblem} reads as follows:

Find $u\in V$ such that
\begin{equation}\label{eq:weakball}
\begin{aligned}
& (\nabla u,\nabla v)_\Om + (\nabla u,\nabla v)_{B_R\setminus\overline\Om}
- \kappa^2(c(\cdot,u) u,v)_{B_R} - (T_\kappa u,v)_{S_R}\\
&= (f(\cdot,u),v)_{B_R}
- (T_\kappa u^\mathrm{inc},v)_{S_R} + (\vecr\cdot\nabla u^\mathrm{inc},v)_{S_R}
\end{aligned}
\end{equation}
for all $v\in V$ holds.

\begin{lemma}\label{l:exintequiv}
The weak formulations \eqref{eq:weakfullspace} and \eqref{eq:weakball}
of the problems \eqref{eq:genscalproblem}--\eqref{eq:sommerfeld1}
and \eqref{eq:intscalproblem}, resp., are equivalent.
\end{lemma}
\begin{proof}First let $u\in V_{\R^d}$ be a weak solution to
\eqref{eq:genscalproblem}--\eqref{eq:sommerfeld1},
i.e.\ it satisfies \eqref{eq:weakfullspace}.
Then its restriction to $B_R$ belongs to $V$.

To demonstrate that this restriction satisfies the weak formulation \eqref{eq:weakball},
we construct the radiating solution $u_{B_{R'}^\mathrm{c}}$
of the homogeneous Helmholtz equation outside of a smaller ball $B_{R'}$
such that $\overline\Om\subset B_{R'}\subset B_R$ and
$\left. u_{B_{R'}^\mathrm{c}}\right|_{S_{R'}} = (u-u^\mathrm{inc})|_{S_{R'}}$.
This solution can be constructed in the form of a series expansion in terms of Hankel functions
as explained in the previous section.
By elliptic regularity (see, e.g., \cite[Thm.~4.16]{McLean:00}, \cite[Sect.~6.3.1]{Evans:15}),
the solution of this problem satisfies the Helmholtz equation in $B_{R'}^\mathrm{c}$.
Moreover, by uniqueness \cite[Thm.~2.6.5]{Nedelec:01},
it coincides with $u-u^\mathrm{inc}=u^\mathrm{rad}$ in $B_{R'}^\mathrm{c}$.

Now we choose a finite partition of unity covering $\overline{B}_R$, denoted by $\{\varphi_j\}_J$
\cite[Sect.~1.2]{Wloka:87}, such that its index set $J$ can be decomposed
into two disjoint subsets $J_1,J_2$ as follows:
\[
\overline{B}_{R'}\subset \inn\Big(\bigcup_{j\in J_1}\supp\varphi_j\Big),
\quad
\bigcup_{j\in J_1}\supp\varphi_j \subset B_R,
\quad
\bigcup_{j\in J_2}\supp\varphi_j \subset B_{R'}^\mathrm{c}.
\]
For example, we can choose $\{\varphi_j\}_{J_1}$ to consist of one element, say $\varphi_1$,
namely the usual mollifier function with support $B'$, where the open ball $B'$ (centered
at the origin) lies between $B_{R'}$ and $B_R$,
i.e.\ $\overline{B}_{R'}\subset B' = \inn\left(\supp\varphi_1\right)$,
$\supp\varphi_1 \subset B_R$.
Then the second part consists of a finite open covering
of the spherical shell $\overline{B}_R\setminus B'$.

Then we take, for any $v\in V$, the product $v_1:=v\sum_{j\in J_1}\varphi_j$.
This is an element of $V$, too, with support in $B_R$, and it can be continued
by zero to an element of $V_{\R^d}^\circ$ (keeping the notation).
Hence we can take it as a test function in the weak formulation \eqref{eq:weakfullspace}
and obtain
\[
a_{\R^d}(u,v_1)=(f(\cdot,u),v_1)_{\R^d}.
\]
This is equal to
\begin{align*}
& (\nabla u,\nabla v_1)_\Om + (\nabla u,\nabla v_1)_{B_R\setminus\overline\Om}
- \kappa^2(c(\cdot,u) u,v_1)_{B_R} - (T_\kappa u,v_1)_{S_R}\\
&= (f(\cdot,u),v_1)_{B_R}
- (T_\kappa u^\mathrm{inc},v_1)_{S_R} + (\vecr\cdot\nabla u^\mathrm{inc},v_1)_{S_R}
\end{align*}
due to the properties of the support of $v_1$
(in particular, all terms ``living'' on $S_R$ are equal to zero).

Since the homogeneous Helmholtz equation is satisfied in
$\bigcup_{j\in J_2}\supp\varphi_j\subset B_{R'}^\mathrm{c}$,
we can proceed as follows. We continue the test function $v_2:=v\sum_{j\in J_2}\varphi_j$
by zero into the complete ball $B_R$ and have
\begin{align*}
(f(\cdot,u),v_2)_{B_R\setminus\overline\Om}&=0
=(-\Delta u - \kappa^2 u,v_2)_{B_R\setminus\overline{B}_{R'}}\\
&=(\nabla u,\nabla v_2)_{B_R\setminus\overline{B}_{R'}}
- \kappa^2 (u, v_2)_{B_R\setminus\overline{B}_{R'}}
- (\vecnu\cdot\nabla u,v_2)_{\partial(B_R\setminus\overline{B}_{R'})}\\
&=(\nabla u,\nabla v_2)_{B_R\setminus\overline{B}_{R'}}
- \kappa^2 (u, v_2)_{B_R\setminus\overline{B}_{R'}} - (\vecr\cdot\nabla u,v_2)_{S_R}.
\end{align*}
Now, taking into consideration the properties of the support of $v_2$,
we easily obtain the following relations:
\begin{align*}
(\nabla u,\nabla v_2)_{B_R\setminus\overline{B}_{R'}}
&=(\nabla u,\nabla v_2)_\Om + (\nabla u,\nabla v_2)_{B_R\setminus\overline\Om},\\
(u, v_2)_{B_R\setminus\overline{B}_{R'}}
&=(c(\cdot,u) u,v_2)_{B_R},\\
(\vecr\cdot\nabla u,v_2)_{S_R}
&=(\vecr\cdot\nabla u^\mathrm{rad},v_2)_{S_R} + (\vecr\cdot\nabla u^\mathrm{inc},v_2)_{S_R}\\
&=(T_\kappa u^\mathrm{rad},v_2)_{S_R} + (\vecr\cdot\nabla u^\mathrm{inc},v_2)_{S_R}\\
&=(T_\kappa u,v_2)_{S_R} - (T_\kappa u^\mathrm{inc},v_2)_{S_R}
+ (\vecr\cdot\nabla u^\mathrm{inc},v_2)_{S_R},
\end{align*}
where the treatment of the last term makes use of the construction
of the Dirichlet-to-Neumann map $T_\kappa$.

Adding both relations and observing that $v=v_1+v_2$,
we arrive at the variational formulation \eqref{eq:weakball}.

Conversely, let $u\in V$ be a solution to \eqref{eq:weakball}.
To continue it into $B_R^\mathrm{c}$, similar to the first part of the proof
we construct the radiating solution $u_{B_R^\mathrm{c}}$
of the Helmholtz equation outside $B_R$ such that
$\left. u_{B_R^\mathrm{c}}\right|_{S_R} = (u-u^\mathrm{inc})|_{S_R}$
and set $u:=u_{B_R^\mathrm{c}} + u^\mathrm{inc}$ in $B_R^+$.
Hence we have that
$T_\kappa u = \frac{\partial u_{B_R^\mathrm{c}}}{\partial\vecr} + T_\kappa u^\mathrm{inc}$.

Now we take an element $v\in V_{\R^d}^\circ$. Its restriction to $B_R$
is an element of $V$ and thus can be taken as a test function in \eqref{eq:weakball}:
\begin{equation}\label{eq:weakball2}
\begin{aligned}
& (\nabla u,\nabla v)_\Om + (\nabla u,\nabla v)_{B_R\setminus\overline\Om}
- \kappa^2(c(\cdot,u) u,v)_{B_R} - (T_\kappa u,v)_{S_R}\\
&= (f(\cdot,u),v)_{B_R}
- (T_\kappa u^\mathrm{inc},v)_{S_R} + (\vecr\cdot\nabla u^\mathrm{inc},v)_{S_R}.
\end{aligned}
\end{equation}
Since $v$ has a compact support, we can choose a ball $B\subset\R^d$ centered at the origin such that
$\overline{B}_R\cup\supp v \subset B$.
The homogeneous Helmholtz equation is obviously satisfied in the spherical shell
$B\setminus\overline{B}_R$:
\[
-\Delta u_{B_R^\mathrm{c}} - \kappa^2 u_{B_R^\mathrm{c}} = 0.
\]
We multiply this equation by the complex conjugate of the test function $v\in V$,
then integrate over the shell, and apply the first Green's formula:
\[
(\nabla u_{B_R^\mathrm{c}},\nabla v)_{B\setminus\overline{B}_R}
-\kappa^2 (u_{B_R^\mathrm{c}},v)_{B\setminus\overline{B}_R}
-(\vecnu\cdot\nabla u_{B_R^\mathrm{c}},v)_{\partial(B\setminus\overline{B}_R)}
= 0.
\]
Now we observe that
\begin{align*}
(\nabla u_{B_R^\mathrm{c}},\nabla v)_{B\setminus\overline{B}_R}
&=(\nabla u_{B_R^\mathrm{c}},\nabla v)_{B_R^+},\\
(u_{B_R^\mathrm{c}},v)_{B\setminus\overline{B}_R}
&=(u_{B_R^\mathrm{c}},v)_{B_R^+},\\
(\vecnu\cdot\nabla u_{B_R^\mathrm{c}},v)_{\partial(B\setminus\overline{B}_R)}
&=-(\vecr\cdot\nabla u_{B_R^\mathrm{c}},v)_{S_R}
=-(T_\kappa u - T_\kappa u^\mathrm{inc},v)_{S_R}
\end{align*}
where the minus sign in the last line results from the change in the orientation
of the outer normal (once w.r.t.\ the shell, once w.r.t. $B_R$)
and the construction of $u_{B_R^\mathrm{c}}$.
So we arrive at
\[ 
(\nabla u_{B_R^\mathrm{c}},\nabla v)_{B_R^+}
-\kappa^2 (u_{B_R^\mathrm{c}},v)_{B_R^+}
+(T_\kappa u,v)_{S_R}
= (T_\kappa u^\mathrm{inc},v)_{S_R}.
\] 
Finally, since the incident field satisfies the homogeneous Helmholtz equation
in the spherical shell, too, we see by an analogous argument that the variational equation
\begin{equation}\label{eq:extball3}
(\nabla u^\mathrm{inc},\nabla v)_{B_R^+}
-\kappa^2 (u^\mathrm{inc},v)_{B_R^+}
= -(\vecr\cdot\nabla u^\mathrm{inc},v)_{S_R}
\end{equation}
holds.

Adding the variational equations \eqref{eq:weakball2} -- \eqref{eq:extball3},
we arrive at the variational formulation \eqref{eq:weakfullspace}.
\end{proof}

\section{Existence and uniqueness of a weak solution}

In this section we investigate the existence and uniqueness of the weak solution
of the interior problem \eqref{eq:intscalproblem}.
We define the sesquilinear form
\begin{equation}\label{eq:ainttransm}
a(w,v)
:=(\nabla w,\nabla v)_\Om + (\nabla w,\nabla v)_{B_R\setminus\overline\Om}
- \kappa^2 (w,v)_{B_R} - (T_\kappa w,v)_{S_R}
\quad\text{for all }
w,v\in V,
\end{equation}
the nonlinear form
\begin{equation}\label{eq:ninttransm}
\begin{aligned}
n(w,v) &:= \kappa^2((c(\cdot,w) - 1) w,v)_{B_R}
+ (f(\cdot,w),v)_{B_R}\\
&\quad
- (T_\kappa u^\mathrm{inc},v)_{S_R}
+ (\vecr\cdot\nabla u^\mathrm{inc},v)_{S_R}
\end{aligned}
\end{equation}
and reformulate \eqref{eq:weakball} as follows:
Find $u\in V$ such that
\begin{equation}\label{eq:intscalproblem3}
a(u,v) = n(u,v)
\quad\text{for all }
v\in V.
\end{equation}
On the space $V$, we use the standard seminorm and norm:
\[
|v|_V := \left(\|\nabla v\|_{0,2,\Om}^2 + \|\nabla v\|_{0,2,B_R\setminus\overline\Om}^2\right)^{1/2},
\quad
\|v\|_V := \left(|v|_V^2 + \|v\|_{0,2,B_R}^2\right)^{1/2}.
\] 
For $\kappa>0$, the following so-called \emph{wavenumber dependent norm} on $V$
is also common:
\[
\|v\|_{V,\kappa} := \left(|v|_V^2 + \kappa^2\|v\|_{0,2,B_R}^2\right)^{1/2}.
\]
It is not difficult to verify that the standard norm and the wavenumber dependent norm
are equivalent on $V$, i.e.\ it holds
\begin{equation}\label{eq:wavenumbernormequiv}
C_- \|v\|_V \le \|v\|_{V,\kappa} \le C_+ \|v\|_V
\quad\text{for all } v\in V,
\end{equation}
where the equivalence constants depend on $\kappa$ in the following way:
$C_- := \min\{1;\kappa\}$ and $C_+ := \max\{1;\kappa\}$.
We now proceed to examine the linear aspects of the problem \eqref{eq:intscalproblem3}.
\begin{lemma}\label{l:a-cont}
The sesquilinear form $a$ is bounded on $V$.
\end{lemma}
\begin{proof}
Applying to each addend in the definition of $a$ the appropriate
Cauchy-Bunyakovsky-Schwarz inequality, we obtain
\begin{align*}
|a(w,v)| &\le |w|_V |v|_V
+ \kappa^2 \|w\|_{0,2,B_R}\|v\|_{0,2,B_R}\\
&\quad + \|T_\kappa w\|_{-1/2,2,S_R} \|v\|_{1/2,2,S_R}
\quad\text{for all } w,v\in V.
\end{align*}
According to Theorem~\ref{th:DtNprop} the DtN operator $T_\kappa$ is bounded,
i.e.\ there exists a constant $C_{T_\kappa} > 0$ such that
\[
\|T_\kappa w\|_{-1/2,2,S_R} \le C_{T_\kappa} \|w\|_{1/2,2,S_R}
\quad\text{for all }
w\in V.
\]
It remains to apply a trace theorem \cite[Thm.~3.37]{McLean:00}:
\begin{align*}
|a(w,v)|
&\le|w|_V |v|_V + \kappa^2 \|w\|_{0,2,B_R}\|v\|_{0,2,B_R}
+ C_{T_\kappa} C_\mathrm{tr}^2
\|w\|_{1,2,B_R\setminus\overline\Om} \|v\|_{1,2,B_R\setminus\overline\Om}\\
&\le |w|_V |v|_V + \kappa^2 \|w\|_{0,2,B_R}\|v\|_{0,2,B_R}
+ C_{T_\kappa} C_\mathrm{tr}^2\|w\|_V \|v\|_V\\
&\le \min\{(\max\{1,\kappa^2\} + C_{T_\kappa} C_\mathrm{tr}^2)\|w\|_V \|v\|_V,
(1 + C_{T_\kappa} C_\mathrm{tr}^2)\|w\|_{V,\kappa} \|v\|_{V,\kappa}\}\\
&\qquad\text{for all } w,v\in V.
\end{align*}
\end{proof}

\begin{lemma}\label{l:contGarding}
Given $\kappa_0>0$ and $R_0>0$,
assume that $\kappa \ge \kappa_0$ (cf.\ Rem.~\ref{rem:sharpDtNbound}) and $R \ge R_0$.
In addition, $\kappa_0\ge 1$ is required for $d=2$.
Then the sesquilinear form a satisfies a G{\aa}rding's inequality of the form
\[
\Re a(v,v) \ge \|v\|_{V,\kappa}^2 - 2\kappa^2 \|v\|_{0,2,B_R}^2
\quad\text{for all } v\in V.
\]
\end{lemma}
\begin{proof}
From the definitions of $a$ and the wavenumber dependent norm it follows immediately that
\begin{align*}
\Re a(v,v) &=\|v\|_{V,\kappa}^2 - 2\kappa^2 \|v\|_{0,2,B_R}^2 - \Re\,(T_\kappa v,v)_{S_R}\\
&\ge \|v\|_{V,\kappa}^2 - 2\kappa^2 \|v\|_{0,2,B_R}^2 + CR^{-1}\|v\|_{0,2,S_R}^2\\
&\ge \|v\|_{V,\kappa}^2 - 2\kappa^2 \|v\|_{0,2,B_R}^2,
\end{align*}
where the first estimate follows from \cite[Lemma~3.3]{Melenk:10} with a constant $C>0$
depending solely on $\kappa_0>0$ and $R_0>0$.
\end{proof}

Next we discuss the solvability and stability of the problem \eqref{eq:intscalproblem3}
for the case that the right-hand side is just an antilinear continuous functional
$\ell:\; V\to\C$.
The linear problem of finding $u\in V$ such that
\begin{equation}\label{eq:intlinproblem}
a(u,v) = \ell(v)
\quad\text{for all } v\in V
\end{equation}
holds can be formulated equivalently as an operator equation in the dual space $V^\ast$ of $V$
consisting of all continuous antilinear functionals from $V$ to $\C$.
Namely, if we define the linear operator $\mathcal{A}:\; V \to V^\ast$ by
\begin{equation}\label{eq:defopA}
\mathcal{A} w(v) := a(w,v)
\quad\text{for all } w, v\in V,
\end{equation}
problem \eqref{eq:intlinproblem} is equivalent to solving the operator equation
\begin{equation}\label{eq:intlinopproblem}
\mathcal{A} u = \ell
\end{equation}
for $u\in V$.

Note that $\mathcal{A}$ is a bounded operator by Lemma~\ref{l:a-cont}.

\begin{theorem}\label{th:linprobsolvable}
Under the assumptions of Lemma~\ref{l:contGarding},
the problem \eqref{eq:intlinopproblem} is uniquely solvable for any $\ell\in V^\ast$.
\end{theorem}
\begin{proof}
The basic ideas of the proof are taken from the proof of \cite[Thm~3.8]{Melenk:10}.
Since the embedding of $V$ into $L_2(B_R)$ is compact by
the compactness theorem of Rellich–Kondrachov \cite[Thm.~3.27]{McLean:00}
together with Tikhonov's product theorem \cite[Thm.~4.1]{Kelley:63},
the compact perturbation theorem \cite[Thm.~2.34]{McLean:00}
together with Lemma~\ref{l:contGarding} imply that
the Fredholm alternative \cite[Thm.~2.27]{McLean:00} holds for the equation \eqref{eq:intlinopproblem}.

Hence it is sufficient to demonstrate that the homogeneous adjoint problem
(cf.\ \cite[p.~43]{McLean:00})
of finding $u\in V$ such that $\overline{a(v,u)} = 0$ holds for all $v\in V$
only allows for the trivial solution.

So suppose $u\in V$ is a solution of the homogeneous adjoint problem.
We take $v:=u$ and consider the imaginary part of the resulting equation:
\[
0 = \Im\overline{a(u,u)}
= - \Im\overline{(T_\kappa u,u)_{S_R}}
= \Im\,(T_\kappa u,u)_{S_R}.
\]
Then \cite[Lemma~3.3]{Melenk:10} implies $u=0$ on $S_R$.
Then $u$ satisfies the variational equation
\[
(\nabla u,\nabla v)_\Om + (\nabla u,\nabla v)_{B_R\setminus\overline\Om}
- \kappa^2(u,v)_{B_R} = 0
\quad\text{for all }
v\in V,
\]
i.e.\ it is a weak solution of the homogeneous interior transmission Neumann problem
for the Helmholtz equation on $B_R$.
On the other hand, $u$ can be extended to the whole space $\R^d$ by zero
to an element $\tilde u\in V_{\R^d}$, and this element can be interpreted
as a weak solution of a homogeneous full-space transmission problem,
for instance in the sense of \cite[Problem (P)]{Torres:93}.
Then it follows from \cite[Lemma~7.1]{Torres:93} that $\tilde u=0$ and thus $u=0$.
\end{proof}

Since a Fredholm operator has a closed image \cite[p.~33]{McLean:00},
it follows from the Open Mapping Theorem and Theorem~\ref{th:linprobsolvable}
(cf.\ \cite[Cor.~2.2]{McLean:00})
that the inverse operator $\mathcal{A}^{-1}$ is bounded,
i.e.\ there exists a constant $C(R,\kappa)>0$ such that
\[
\|u\|_{V,\kappa} = \|\mathcal{A}^{-1}\ell\|_{V,\kappa} \le C(R,\kappa) \|\ell\|_{V^\ast}
\quad\text{for all }
\ell\in V^\ast.
\]
Then it holds
\[
\frac{1}{C(R,\kappa)}\le \frac{\|\ell\|_{V^\ast}}{\|u\|_{V,\kappa}}
=\sup_{v\in V\setminus\{0\}}\frac{|\ell(v)|}{\|u\|_{V,\kappa}\|v\|_{V,\kappa}}
=\sup_{v\in V\setminus\{0\}}\frac{|a(u,v)|}{\|u\|_{V,\kappa}\|v\|_{V,\kappa}}\,.
\]
This estimate proves the following result.
\begin{lemma}\label{l:infsup}
Under the assumptions of Lemma~\ref{l:contGarding}, the sesquilinear form
$a$ satisfies an \emph{inf-sup condition}:
\[
\beta(R,\kappa) := \inf_{w\in V\setminus\{0\}}\sup_{v\in V\setminus\{0\}}
\frac{|a(w,v)|}{\|w\|_{V,\kappa} \|v\|_{V,\kappa}} > 0.
\]
\end{lemma}

Now we turn to the nonlinear situation and concretize the assumptions
regarding the Cara\-th\'{e}o\-dory functions $c$ and $f$.

\begin{lemma}\label{l:f-est}
Let $p_f\in\begin{cases}
[2,\infty), & d=2,\\
[2,6], & d=3,
\end{cases}$
and assume there exist nonnegative functions $m_f,g_f\in L_\infty(\Om)$
such that
\[
|f(\vecx,\xi)|\le m_f(\vecx)|\xi|^{p_f-1} + g_f(\vecx)
\quad\text{for all }
(\vecx,\xi)\in \Om\times\C.
\]
Then $vf(\cdot,w) \in L_1(\Om)$ for all $w,v\in V$.
\end{lemma}
\begin{proof}
Since $f$ is a Carath\'{e}odory function, the composition $f(\cdot,w)$ is measurable
and it suffices to estimate the integral of $|vf(\cdot,w)|$.
Moreover, it suffices to consider the term $m_f v |w|^{p_f-1}$ in more detail.
By H\"{o}lder's inequality for three functions, it holds that
\[
\|vf(\cdot,w)\|_{0,1,\Om}
\le \|m_f\|_{0,\infty,\Om}\|v\|_{0,p_f,\Om}
\|w^{p_f-1}\|_{0,q,\Om}
\quad\text{with }
\frac{1}{p_f}+\frac{1}{q}=1.
\]
The $L_{p_f}$-norm of $v$ is bounded thanks to the embedding
$V|_\Om\subset L_{p_f}(\Om)$ for the allowed values of $p_f$ \cite[Thm.~4.12]{Adams:03}.
Since $|w^{p_f-1}|^q=|w|^{p_f}$, the $L_q$-norm of $w^{p_f-1}$ is bounded by the same reasoning.
\end{proof}

\begin{lemma}\label{l:c-est}
Let $p_c\in\begin{cases}
[2,\infty), & d=2,\\
[2,6], & d=3,
\end{cases}$
and assume there exist nonnegative functions $m_c,g_c\in L_\infty(\Om)$
such that
\[
|c(\vecx,\xi)-1|\le m_c(\vecx)|\xi|^{p_c-2} + g_c(\vecx)
\quad\text{for all }
(\vecx,\xi)\in \Om\times\C.
\]
Then $zv(c(\cdot,w)-1)\in L_1(\Om)$ for all $z,w,v\in V$.
\end{lemma}
\begin{proof}
Similar to the proof of Lemma~\ref{l:f-est} it is sufficient
to consider the term $m_c z v |w|^{p_c-2}$ in more detail.
By H\"{o}lder's inequality for four functions, it holds that
\[
\|zv(c(\cdot,w)-1)\|_{0,1,\Om}
\le \|m_c\|_{0,\infty,\Om}\|z\|_{0,p_c,\Om}\|v\|_{0,p_c,\Om}\|w^{p_c-2}\|_{0,q,\Om}
\quad\text{with }
\frac{2}{p_c}+\frac{1}{q}=1.
\]
The $L_{p_c}$-norms of $z,v$ are bounded thanks to the embedding theorem \cite[Thm.~4.12]{Adams:03}.
Since $|w^{p_c-2}|^q=|w|^p_c$, the $L_q$-norm of $w^{p_c-2}$ is bounded by the same reasoning.
\end{proof}
\begin{corollary}\label{cor:cf-inequal}
Under the assumptions of Lemma~\ref{l:f-est} and Lemma~\ref{l:c-est}, resp.,
the following estimates hold for all $z,w,v\in V$:
\begin{align*}
|(f(\cdot,w),v)_\Om|
& \le C_\mathrm{emb}^{p_f}\|m_f\|_{0,\infty,\Om}\|w\|_{1,2,\Om}^{p_f-1}\|v\|_{1,2,\Om} \\
&\quad + \sqrt{|\Om|_d}\,\|g_f\|_{0,\infty,\Om}\|v\|_{0,2,\Om},\\
|((c(\cdot,w)-1)z,v)_\Om|
& \le C_\mathrm{emb}^{p_c}\|m_c\|_{0,\infty,\Om}\|w\|_{1,2,\Om}^{p_c-2}\|z\|_{1,2,\Om}\|v\|_{1,2,\Om} \\
&\quad + \|g_c\|_{0,\infty,\Om}\|z\|_{0,2,\Om}\|v\|_{0,2,\Om},
\end{align*}
where $|\Om|_d$ is the $d$-volume of $\Om$.
\end{corollary}
\begin{proof}
Replace $v$ by $\overline{v}$ in Lemmata~\ref{l:f-est}, \ref{l:c-est}
to get the first addend of the bounds.
The estimate of the second addend is trivial.
\end{proof}

\begin{example}\label{ex:Kerr}
An important example for the nonlinearities is
\[
c(\vecx,\xi) := \begin{cases}
1, & (\vecx,\xi)\in \Om^+\times\C,\\
\eps^{(L)}(\vecx) + \alpha(\vecx)|\xi|^2, & (\vecx,\xi)\in\Om\times\C,
\end{cases}
\]
with given $\eps^{(L)}, \alpha\in L_\infty(\Om)$,
and $f = 0$. Here $p_c=4$, which is within the range of validity of Lemma~\ref{l:c-est},
and $m_c=|\alpha|$, $g_c=|\eps^{(L)}-1|$.
\end{example}

The estimates from Corollary~\ref{cor:cf-inequal} show that the first
two terms on the right-hand side of the variational equation \eqref{eq:intscalproblem3}
can be considered as values of nonlinear mappings from $V$ to $V^\ast$,
i.e.\ we can define
\begin{align*}
\ell^\mathrm{contr}:\; V\to V^\ast
\quad\text{ by}&
&\langle\ell^\mathrm{contr}(w),v\rangle &:= \kappa^2((c(\cdot,w) - 1) w,v)_\Om,\\
\ell^\mathrm{src}:\; V\to V^\ast
\quad\text{ by}&
&\langle\ell^\mathrm{src}(w),v\rangle &:=  (f(\cdot,w),v)_\Om
\qquad\text{for all }
w,v\in V.
\end{align*}
Furthermore,
if $u^\mathrm{inc}\in H^1_\mathrm{loc}(\Om^+)$ is such that
additionally $\Delta u^\mathrm{inc}$ belongs to $L_{2,\mathrm{loc}}(\Om^+)$
(where $\Delta u^\mathrm{inc}$ is understood in the distributional sense),
the last two terms on the right-hand side of \eqref{eq:ninttransm}
form an antilinear continuous functional on $\ell^\mathrm{inc}\in V^\ast$:
\[
\langle\ell^\mathrm{inc},v\rangle
:=(\vecr\cdot\nabla u^\mathrm{inc} - T_\kappa u^\mathrm{inc},v)_{S_R}
\qquad\text{for all }
v\in V.
\]
This is a consequence of Theorem~\ref{th:DtNprop} and the estimates before the trace theorem
\cite[Thm.~6.13]{Angermann:21a}.
Hence
\[
\|\ell^\mathrm{inc}\|_{V^\ast}
\le \tilde{C}_\mathrm{tr}[\|\Delta u^\mathrm{inc}\|_{0,2,B_R\setminus\overline{\Om}}
+ \|u^\mathrm{inc}\|_{0,2,B_R\setminus\overline{\Om}}]
+ C_{T_\kappa} C_\mathrm{tr}^2\|u^\mathrm{inc}\|_{1,2,B_R\setminus\overline{\Om}},
\]
where $\tilde{C}_\mathrm{tr}$ is the norm of the trace operator defined in
\cite[eq.~(6.39)]{Angermann:21a}.

However, it is more intuitive to utilize the estimate
\[
\|\ell^\mathrm{inc}\|_{V^\ast}
\le C_\mathrm{tr}\|\vecr\cdot\nabla u^\mathrm{inc} - T_\kappa u^\mathrm{inc}\|_{-1/2,2,S_R}.
\]
The reason for this is that the bound can be interpreted as a measure
of the deviation of the function $u^\mathrm{inc}$ from a radiating solution
of the corresponding Helmholtz equation.
In other words: If the function $u^\mathrm{inc}$ satisfies the boundary value problem
\eqref{eq:extDirproblem} with $f_{S_R}:=u^\mathrm{inc}|_{S_R}$,
then the functional $\ell^\mathrm{inc}$ is not present.

Consequently, setting
\[
\mathcal{F}(w) := \ell^\mathrm{contr}(w) + \ell^\mathrm{src}(w) + \ell^\mathrm{inc}
\qquad\text{for all }
w\in V,
\]
we obtain a nonlinear operator $\mathcal{F}:\; V \to V^\ast$,
and the problem \eqref{eq:intscalproblem3} is then equivalent to the operator equation
\[
\mathcal{A} u = \mathcal{F}(u)
\quad\text{in } V^\ast,
\]
and further, by Lemma~\ref{l:infsup}, equivalent to the fixed-point problem
\begin{equation}\label{eq:scalfpp}
u = \mathcal{A}^{-1} \mathcal{F}(u)
\quad\text{in } V.
\end{equation}
In order to prove the subsequent existence and uniqueness theorem, we specify
some additional properties of the nonlinearities $c$ and $f$.
\begin{definition}\label{def:loclip}
The functions $c$ and $f$ are said to generate locally Lipschitz continuous
Nemycki operators in $V$ if the following holds: For some parameters
$p_c,p_f\in\begin{cases}
[2,\infty), & d=2,\\
[2,6], & d=3,
\end{cases}$,
there exist Carath\'{e}odory functions
$L_c:\;\Om\times\C\times\C\to(0,\infty)$ and $L_f:\;\Om\times\C\times\C\to(0,\infty)$
such that the composition operators
$\Om\times V\times V\to L_{q_c}(\Om):\:(\vecx,w,v)\mapsto L_c(\vecx,w,v)$,
$\Om\times V\times V\to L_{q_f}(\Om):\:(\vecx,w,v)\mapsto L_f(\vecx,w,v)$
are bounded for $q_c,q_f>0$ with
$\frac{3}{p_c}+\frac{1}{q_c}=\frac{2}{p_f}+\frac{1}{q_f}=1$,
and
\[
|c(\vecx,\xi)-c(\vecx,\eta)| \le L_c(\vecx,\xi,\eta) |\xi-\eta|,
\quad
|f(\vecx,\xi)-f(\vecx,\eta)| \le L_f(\vecx,\xi,\eta) |\xi-\eta|
\]
for all $(\vecx,\xi,\eta)\in\Om\times\C\times\C$.
\end{definition}

\begin{remark}\label{rem:df-est}
If the nonlinearities $c$ and $f$ generate locally Lipschitz continuous
Nemycki operators in the sense of the above Definition~\ref{def:loclip}, the
assumptions of Lemmata~\ref{l:f-est}, \ref{l:c-est} can be replaced by
the requirement that there exist functions $w_f,w_c\in V$ such that
$f(\cdot,w_f)\in L_{p_f/(p_f-1)}(\Om)$ and $c(\cdot,w_c)\in L_{p_c/(p_c-2)}(\Om)$,
respectively.
\end{remark}
\begin{proof}
Indeed, similar to the proofs of the two lemmata mentioned, we have that
\begin{align*}
\|vf(\cdot,w)\|_{0,1,\Om}
&\le \|vf(\cdot,w_f)\|_{0,1,\Om} + \|v(f(\cdot,w) - f(\cdot,w_f))\|_{0,1,\Om}\\
&\le \|vf(\cdot,w_f)\|_{0,1,\Om} + \|vL_f(\cdot,w,w_f) |w-w_f|\|_{0,1,\Om}\\
&\le \|v\|_{0,p_f,\Om}\|f(\cdot,w_f)\|_{0,\tilde{q}_f,\Om}
+ \|v\|_{0,p_f,\Om}\|L_f(\cdot,w,w_f)\|_{0,q_f,\Om}\|w-w_f\|_{0,p_f,\Om}\\
&\le \left[\|f(\cdot,w_f)\|_{0,\tilde{q}_f,\Om}
+ \|L_f(\cdot,w,w_f)\|_{0,q_f,\Om}(\|w\|_V + \|w_f\|_V)\right]\|v\|_V,\\
\|zvc(\cdot,w)\|_{0,1,\Om}
&\le \|zvc(\cdot,w_c)\|_{0,1,\Om} + \|zv(c(\cdot,w) - c(\cdot,w_c))\|_{0,1,\Om}\\
&\le \|zvc(\cdot,w_c)\|_{0,1,\Om} + \|zvL_c(\cdot,w,w_c) |w-w_c|\|_{0,1,\Om}\\
&\le \|z\|_{0,p_c,\Om}\|v\|_{0,p_c,\Om}\|c(\cdot,w_c)\|_{0,\tilde{q}_c,\Om}\\
&\quad + \|z\|_{0,p_c,\Om}\|v\|_{0,p_c,\Om}\|L_c(\cdot,w,w_c)\|_{0,q_c,\Om}\|w-w_c\|_{0,p_c,\Om}\\
&\le \left[\|c(\cdot,w_c)\|_{0,\tilde{q}_c,\Om}
+ \|L_c(\cdot,w,w_c)\|_{0,q_c,\Om}(\|w\|_V + \|w_c\|_V)\right]\|z\|_V\|v\|_V
\end{align*}
with
$\frac{1}{p_f}+\frac{1}{\tilde{q}_f}=1$
and
$\frac{2}{p_c}+\frac{1}{\tilde{q}_c}=1$.
\end{proof}

\begin{theorem}\label{th:scalee}
Under the assumptions of Lemma~\ref{l:contGarding},
let the functions $c$ and $f$ generate locally Lipschitz continuous
Nemycki operators in $V$ and assume that there exist functions $w_f,w_c\in V$ such that
$f(\cdot,w_f)\in L_{p_f/(p_f-1)}(\Om)$ and $c(\cdot,w_f)\in L_{p_c/(p_c-2)}(\Om)$,
respectively.

Furthermore let $u^\mathrm{inc}\in H^1_\mathrm{loc}(\Om^+)$ be such that
additionally $\Delta u^\mathrm{inc}\in L_{2,\mathrm{loc}}(\Om^+)$ holds.

If there exist numbers $\varrho>0$ and $L_\mathcal{F}\in (0,\beta(R,\kappa))$ such that
the following two conditions
\begin{align}
&\kappa^2\left[\|c(\cdot,w_c)-1\|_{0,\tilde{q}_c,\Om}
+ \|L_c(\cdot,w,w_c)\|_{0,q_c,\Om}(\varrho + \|w_c\|_V)\right]\varrho\nonumber\\
&\quad + \left[\|f(\cdot,w_f)\|_{0,\tilde{q}_f,\Om}
+ \|L_f(\cdot,w,w_f)\|_{0,q_f,\Om}(\varrho + \|w_f\|_V)\right]\label{eq:theu1}\\
&\quad
+ C_\mathrm{tr}\|\vecr\cdot\nabla u^\mathrm{inc} - T_\kappa u^\mathrm{inc}\|_{-1/2,2,S_R}
\le \varrho\beta(R,\kappa),\nonumber\\
&\kappa^2\left[\|L_c(\cdot,w,v)\|_{0,q_c,\Om}\varrho
+ \|c(\cdot,w_c)-1\|_{0,\tilde{q}_c,\Om}
+ \|L_c(\cdot,w,w_c)\|_{0,q_c,\Om}(\varrho + \|w_c\|_V)\right]\nonumber\\
&\quad + \|L_f(\cdot,w,v)\|_{0,q_f,\Om} \le L_\mathcal{F}
\label{eq:theu2}
\end{align}
are satisfied for all $w,v\in K_\varrho^\mathrm{cl}:=\{v \in V:\; \|v\|_V \le \varrho\}$,
then the problem \eqref{eq:scalfpp} has a unique solution $u \in K_\varrho^\mathrm{cl}$.
\end{theorem}
\begin{proof}
First we mention that $K_\varrho^\mathrm{cl}$ is a closed nonempty subset of $V$.

Next we show that
$\mathcal{A}^{-1}\mathcal{F}(K_\varrho^\mathrm{cl}) \subset K_\varrho^\mathrm{cl}$.
To this end we make use of the estimates given in the proof of Remark~\ref{rem:df-est}
and obtain
\begin{align*}
\|\mathcal{F}(w)\|_{V^\ast}
&\le \|\ell^\mathrm{contr}(w)\|_{V^\ast} + \|\ell^\mathrm{src}(w)\|_{V^\ast}
+\|\ell^\mathrm{inc}\|_{V^\ast}\\
&\le \kappa^2\left[\|c(\cdot,w_c)-1\|_{0,\tilde{q}_c,\Om}
+ \|L_c(\cdot,w,w_c)\|_{0,q_c,\Om}(\|w\|_V + \|w_c\|_V)\right]\|w\|_V\\
&\quad + \left[\|f(\cdot,w_f)\|_{0,\tilde{q}_f,\Om}
+ \|L_f(\cdot,w,w_f)\|_{0,q_f,\Om}(\|w\|_V + \|w_f\|_V)\right]
+ \|\ell^\mathrm{inc}\|_{V^\ast}\\
&\le \kappa^2\left[\|c(\cdot,w_c)-1\|_{0,\tilde{q}_c,\Om}
+ \|L_c(\cdot,w,w_c)\|_{0,q_c,\Om}(\varrho + \|w_c\|_V)\right]\varrho\\
&\quad + \left[\|f(\cdot,w_f)\|_{0,\tilde{q}_f,\Om}
+ \|L_f(\cdot,w,w_f)\|_{0,q_f,\Om}(\varrho + \|w_f\|_V)\right]\\
&\quad
+ C_\mathrm{tr}\|\vecr\cdot\nabla u^\mathrm{inc} - T_\kappa u^\mathrm{inc}\|_{-1/2,2,S_R}\,.
\end{align*}

\iffalse
\eqref{eq:f-est}%
--\eqref{eq:linincpartnorm}
and obtain
\begin{align*}
\|\mathcal{F}(w)\|_{V^\ast}
&\le \|\ell^\mathrm{contr}(w)\|_{V^\ast} + \|\ell^\mathrm{src}(w)\|_{V^\ast}
+\|\ell^\mathrm{inc}\|_{V^\ast}\\
&\le \kappa^2\left[C_\mathrm{emb}^2\|m_c\|_{0,\infty,\Om}\|w\|_{1,2,\Om}^{p-1}
+ \|g_c\|_{0,\infty,\Om}\|w\|_{0,2,\Om}\right]\\
&\quad + C_\mathrm{emb}\|m_f\|_{0,\infty,\Om}\|w\|_{1,2,\Om}^{p-1}
+ \sqrt{|\Om|_d}\,\|g_f\|_{0,\infty,\Om}
+ \|\ell^\mathrm{inc}\|_{V^\ast}\\
&\le \kappa^2[C_\mathrm{emb}^2\|m_c\|_{0,\infty,\Om}\varrho^{p-1}
+ \|g_c\|_{0,\infty,\Om}\varrho]
+ C_\mathrm{emb}\|m_f\|_{0,\infty,\Om}\varrho^{p-1}
+ \sqrt{|\Om|_d}\,\|g_f\|_{0,\infty,\Om}\\
&\quad
+ C_\mathrm{tr}\|\vecr\cdot\nabla u^\mathrm{inc} - T_\kappa u^\mathrm{inc}\|_{-1/2,2,S_R}.
\end{align*}
\fi
Hence the assumption \eqref{eq:theu1} implies $\|\mathcal{A}^{-1}\mathcal{F}(w)\|_V \le \varrho$.

It remains to show that the mapping $\mathcal{A}^{-1}\mathcal{F}$ is a contraction.

We start with the consideration of the contrast term. From the elementary
decomposition
\[
(c(\cdot,w) - 1) w - (c(\cdot,v) - 1) v
=(c(\cdot,w) - c(\cdot,v)) w + (c(\cdot,v)-1) (w - v)
\]
we see that
\begin{align*}
&\|\ell^\mathrm{contr}(w) - \ell^\mathrm{contr}(v)\|_{V^\ast}\\
&\le \kappa^2\|L_c(\cdot,w,v)\|_{0,q_c,\Om}\|w-v\|_V\|w\|_V\\
&\quad + \kappa^2\left[\|c(\cdot,w_c)-1\|_{0,\tilde{q}_c,\Om}
+ \|L_c(\cdot,w,w_c)\|_{0,q_c,\Om}\|w-w_c\|_V\right]\|w-v\|_V\\
&\le \kappa^2\|L_c(\cdot,w,v)\|_{0,q_c,\Om}\|w-v\|_V\varrho\\
&\quad + \kappa^2\left[\|c(\cdot,w_c)-1\|_{0,\tilde{q}_c,\Om}
+ \|L_c(\cdot,w,w_c)\|_{0,q_c,\Om}(\varrho + \|w_c\|_V)\right]\|w-v\|_V\\
&\le \kappa^2\left[\|L_c(\cdot,w,v)\|_{0,q_c,\Om}\varrho
+ \|c(\cdot,w_c)-1\|_{0,\tilde{q}_c,\Om}
+ \|L_c(\cdot,w,w_c)\|_{0,q_c,\Om}(\varrho + \|w_c\|_V)\right]\|w-v\|_V.
\end{align*}
The estimate of the source term follows immediately from the properties of $f$:
\[
\|\ell^\mathrm{src}(w) - \ell^\mathrm{src}(v)\|_{V^\ast}
\le \|L_f(\cdot,w,v)\|_{0,q_f,\Om}\|w-v\|_V.
\]
From
\[
\|\mathcal{F}(w) - \mathcal{F}(v)\|_{V^\ast}
\le \|\ell^\mathrm{contr}(w) - \ell^\mathrm{contr}(v)\|_{V^\ast}
+ \|\ell^\mathrm{src}(w) - \ell^\mathrm{src}(v)\|_{V^\ast}
\]
and assumption \eqref{eq:theu2} we thus obtain
\[
\|\mathcal{F}(w) - \mathcal{F}(v)\|_{V^\ast}
\le L_\mathcal{F} \|w-v\|_V.
\]
In summary, Banach's fixed point theorem can be applied (see e.g.\ \cite[Sect.~9.2.1]{Evans:15})
and we conclude that the problem \eqref{eq:scalfpp}
has a unique solution $u\in$ $K_\varrho^\mathrm{cl}$.
\end{proof}

If we introduce the function space
\[
\tilde{V} := \{v\in L_2(B_R):\; v|_\Om\in H^1(\Om),\ v|_{B_R\setminus\overline\Om}\in H^1(B_R\setminus\overline\Om)\}
\]
equipped with the norm
\begin{align*}
\|v\|_{\tilde{V}}
& := \left(\|v\|_{1,2,\Om}^2 + \|v\|_{1,2,B_R\setminus\overline\Om}^2\right)^{1/2}
\quad
\text{for all }
v \in \tilde{V},
\end{align*}
the ball $K_\varrho^\mathrm{cl}$ appearing in the above theorem
can be interpreted as a ball in $\tilde{V}$ of radius $\varrho$ with center in
\[
u_0 := \begin{cases}
0 & \text{in } \Om,\\
-u^\mathrm{inc} & \text{in } B_R\setminus\overline\Om.
\end{cases}
\]
Indeed, for $u$ of the form \eqref{eq:solstructure}, it holds that
\begin{align*}
\|u-u_0\|_{\tilde{V}}^2
= \|u^\mathrm{trans}\|_{1,2,\Om}^2
+ \|u^\mathrm{rad}+u^\mathrm{inc}\|_{1,2,B_R\setminus\overline\Om}^2
= \|u\|_V^2.
\end{align*}
This means that the influence of the incident field $u^\mathrm{inc}$ on the radius $\varrho$
in Theorem~\ref{th:scalee} depends only on the deviation of $u^\mathrm{inc}$
from a radiating field measured by $\|\ell^\mathrm{inc}\|_{V^\ast}$,
but not directly on the intensity of $u^\mathrm{inc}$.
In other words, if the incident field $u^\mathrm{inc}$ is radiating
(i.e., it also satisfies the Sommerfeld radiation condition \eqref{eq:sommerfeld1}
and thus $\ell^\mathrm{inc}=0$),
the radius $\varrho$ does not depend on $u^\mathrm{inc}$.
In particular, $u^\mathrm{inc}$ can be a strong field, which is important
for the occurrence of generation effects of higher harmonics \cite{Angermann:19a}.

\begin{example}[Example \ref{ex:Kerr} continued]\label{ex:Kerr_cont}
The identity
\[
c(\cdot,\xi)-c(\cdot,\eta) = \alpha\,(|\xi|^2 - |\eta|^2)
= \alpha\,(|\xi| + |\eta|)(|\xi| - |\eta|)
\]
for all $\xi,\eta\in\C$
and the inequality
$
\left||\xi| - |\eta|\right| \le |\xi - \eta|
$
show that
\[
|c(\cdot,\xi)-c(\cdot,\eta)|
\le |\alpha|(|\xi| + |\eta|) |\xi - \eta|
\]
holds, hence we can set
$L_c(\cdot,\xi,\eta) := |\alpha|(|\xi| + |\eta|)$.
With $p_c=q_c=4$, $c$ generates a locally Lipschitz continuous Nemycki operator
in $V$.
Furthermore we may choose $w_c=0$.
Then:
\begin{align*}
\|c(\cdot,w_c)-1\|_{0,\tilde{q}_c,\Om}
&= \|\eps^{(L)} - 1\|_{0,2,\Om},\\
\|L_c(\cdot,w,v)\|_{0,q_c,\Om}
&= \|\alpha(|w| + |v|)\|_{0,4,\Om}
\le \|\alpha w\|_{0,4,\Om} + \|\alpha v\|_{0,4,\Om}\\
&\le  \|\alpha\|_{0,\infty,\Om}\left[\|w\|_{0,4,\Om} + \|v\|_{0,4,\Om}\right]
\le  C_\mathrm{emb}\|\alpha\|_{0,\infty,\Om}\left[\|w\|_V + \|v\|_V\right],\\
\|L_c(\cdot,w,w_c)\|_{0,q_c,\Om}
&=  \|\alpha w\|_{0,4,\Om}
\le  C_\mathrm{emb}\|\alpha\|_{0,\infty,\Om}\|w\|_V.
\end{align*}
Hence the validity of the following conditions is sufficient for
\eqref{eq:theu1}, \eqref{eq:theu2}:
\begin{align*}
\kappa^2\left[\|\eps^{(L)} - 1\|_{0,2,\Om}
+ C_\mathrm{emb}\|\alpha\|_{0,\infty,\Om}\varrho^2\right]\varrho&\\
+\ C_\mathrm{tr}\|\vecr\cdot\nabla u^\mathrm{inc} - T_\kappa u^\mathrm{inc}\|_{-1/2,2,S_R}
&\le \varrho\beta(R,\kappa),\\
\kappa^2\left[\|\eps^{(L)} - 1\|_{0,2,\Om}
+ 3C_\mathrm{emb}\|\alpha\|_{0,\infty,\Om}\varrho^2\right]
&\le L_\mathcal{F}.
\end{align*}
A consideration of these condition shows that there can be different scenarios
for which they can be fulfilled.
In particular, one of the smallness requirements concerns the product
$\|\alpha\|_{0,\infty,\Om}\varrho^3$.
\end{example}
\begin{example}[saturated Kerr nonlinearity]
Another important example for the nonlinearities is \cite{Akhmediev:98}
\[
c(\vecx,\xi) := \begin{cases}
1, & (\vecx,\xi)\in \Om^+\times\C,\\
\eps^{(L)}(\vecx) + \alpha(\vecx)|\xi|^2/(1+\gamma|\xi|^2), & (\vecx,\xi)\in\Om\times\C,
\end{cases}
\]
with given $\eps^{(L)}, \alpha\in L_\infty(\Om)$, saturation parameter $\gamma>0$,
and $f = 0$. Based on the identity
\[
\frac{|\xi|^2}{1+\gamma|\xi|^2} - \frac{|\eta|^2}{1+\gamma|\eta|^2}
= \frac{(1+\gamma|\eta|^2)|\xi|^2-(1+\gamma|\xi|^2)|\eta|^2}{(1+\gamma|\xi|^2)(1+\gamma|\eta|^2)}
=\frac{|\xi|^2-|\eta|^2}{(1+\gamma|\xi|^2)(1+\gamma|\eta|^2)}
\]
for all $\xi,\eta\in\C$ we obtain
\[
\left|\frac{|\xi|^2}{1+\gamma|\xi|^2} - \frac{|\eta|^2}{1+\gamma|\eta|^2}\right|
= \frac{(|\xi| + |\eta|)\left||\xi| - |\eta|\right|}{(1+\gamma|\xi|^2)(1+\gamma|\eta|^2)}
\le (|\xi| + |\eta|)|\xi - \eta|.
\]
Hence on $\Om$ we arrive at the same Lipschitz function as in the previous
Example~\ref{ex:Kerr_cont}, that is
\[
L_c(\vecx,\xi,\eta) := \begin{cases}
0, & (\vecx,\xi,\eta)\in \Om^+\times\C\times\C,\\
|\alpha|(|\xi| + |\eta|), & (\vecx,\xi,\eta)\in\Om\times\C\times\C.
\end{cases}
\]
Moreover, since
\[
c(\vecx,w_c) = c(\vecx,0) = \begin{cases}
1, & (\vecx,\xi)\in \Om^+\times\C,\\
\eps^{(L)}, & (\vecx,\xi)\in\Om\times\C.
\end{cases}
\]we get the same sufficient conditions.
\end{example}

\section{The modified boundary value problem}

Since the exact DtN operator is represented as an infinite series
(see \eqref{def:2dDtNb}, \eqref{def:3dDtNb}),
it is practically necessary to truncate this nonlocal operator
and consider only finite sums
\begin{align}
T_{\kappa,N} u(\vecx)
&:= \frac{1}{R}\sum_{|n|\le N} Z_n(\kappa R)u_n(R)Y_n(\hat\vecx),
\quad
\vecx=R\hat\vecx\in S_R\subset\R^2,
\label{def:2dDtNtrunc}\\
T_{\kappa,N} u(\vecx)
&=\frac{1}{R}\sum_{n=0}^N\sum_{|m|\le n}
z_n(\kappa R)u_n^m(R)Y_n^m(\hat\vecx),
\quad
\vecx=R\hat\vecx\in S_R\subset\R^3
\label{def:3dDtNtrunc}
\end{align}
for some $N\in\N_0$. The map $T_{\kappa,N}$ is called the \emph{truncated DtN operator},
and $N$ is the \emph{truncation order} of the DtN operator.

The replacement of the exact DtN operator $T_\kappa$ in the problem \eqref{eq:intscalproblem}
by the truncated DtN operator $T_{\kappa,N}$ introduces a perturbation,
hence we have to answer the question of existence and uniqueness of a solution
to the following problem:

Find $u_N\in V$ such that
\begin{equation}\label{eq:intscalproblem4}
a_N(u_N,v) = n_N(u_N,v)
\quad\text{for all }
v\in V.
\end{equation}
holds, where $a_N$ and $n_N$ are the forms defined by \eqref{eq:ainttransm},
\eqref{eq:ninttransm} with $T_\kappa$ replaced by $T_{\kappa,N}$.

The next result is the counterpart to Lemmata \ref{l:a-cont}, \ref{l:contGarding}.
Here we formulate a different version of G{\aa}rding's inequality compared to the case
$d = 2$ considered in \cite[Thm.~4.4]{Hsiao:11}.
\begin{lemma}\label{l:aNprop}
The sesquilinear form $a_N$
\begin{enumerate}[(i)]
\item
is bounded, i.e.\ there exists a constant $C > 0$ independent of $N$ such that
\[
|a_N(w,v)| \le C\|w\|_V \|v\|_V
\quad\text{for all } w,v\in V,
\]
and
\item
satisfies a G{\aa}rding's inequality in the form
\[
\Re a_N(v,v) \ge \|v\|_{V,\kappa}^2 - 2\kappa^2 \|v\|_{0,2,B_R}^2
\quad\text{for all } v\in V.
\]
\end{enumerate}
\end{lemma}
\begin{proof}
(i)
If the proof of \cite[eq.~(3.4a)]{Melenk:10} is carried out with finitely many terms
of the expansion of $T_\kappa$ only, the statement follows easily.
Alternatively, Lemma~\ref{l:Tkuniformcontinuity} with $s=0$ can also be used.

(ii)
As in the proof of Lemma~\ref{l:contGarding},
the definitions of $a_N$ and the wavenumber dependent norm yield
\[
\Re a_N(v,v) =\|v\|_{V,\kappa}^2 - 2\kappa^2 \|v\|_{0,2,B_R}^2 - \Re\,(T_{\kappa,N} v,v)_{S_R}.
\]
Hence it remains to estimate the last term.
In the case $d=2$, we have (see \eqref{def:2dDtNtrunc})
\[
T_{\kappa,N} v(\vecx)
:= \frac{1}{R}\sum_{|n|\le N} Z_n(\kappa R)v_n(R)Y_n(\hat\vecx),
\quad
\vecx=R\hat\vecx\in S_R.
\]
Then, using the $L_2(S_1)$-orthonormality of the circular harmonics
\cite[Prop.~3.2.1]{Zeidler:95a}, we get
\begin{align*}
-(T_{\kappa,N} v,v)_{S_R}
&= -\frac{1}{R}\sum_{|n|\le N} Z_n(\kappa R)(v_n(R)Y_n,v_n(R)Y_n)_{S_R}\\
&= -\frac{1}{R}\sum_{|n|\le N} Z_n(\kappa R)|v_n(R)|^2(Y_n,Y_n)_{S_R}\\
&= -\sum_{|n|\le N} Z_n(\kappa R)|v_n(R)|^2(Y_n,Y_n)_{S_1}\\
&= -\sum_{|n|\le N} Z_n(\kappa R)|v_n(R)|^2.
\end{align*}
Hence, by Lemma~\ref{l:propHankelfunctions},
\begin{align*}
-\Re\,(T_{\kappa,N} v,v)_{S_R}
&=\sum_{0<|n|\le N} \underbrace{(-\Re\,Z_n(\kappa R))}_{\ge 1/2}|v_n(R)|^2
+ \underbrace{(-\Re\,Z_0(\kappa R))}_{> 0}|v_0(R)|^2\\
&\ge \frac{1}{2}\sum_{|n|\le N} |v_n(R)|^2 \ge 0.
\end{align*}
The case $d=3$ can be treated similarly. From
\[
T_{\kappa,N} v(\vecx)
=\frac{1}{R}\sum_{n=0}^N\sum_{|m|\le n}
z_n(\kappa R)v_n^m(R)Y_n^m(\hat\vecx)
\]
(see \eqref{def:3dDtNtrunc}),
we immediately obtain, using the $L_2(S_1)$-orthonormality of the spherical harmonics
\cite[Thm.~2.8]{Colton:19} that
\begin{align*}
-(T_{\kappa,N} v,v)_{S_R}
&=-\frac{1}{R}\sum_{n=0}^N\sum_{|m|\le n}
z_n(\kappa R)(v_n^m(R)Y_n^m,v_n^m(R)Y_n^m)_{S_R}\\
&=-\frac{1}{R}\sum_{n=0}^N\sum_{|m|\le n}
z_n(\kappa R)|v_n^m(R)|^2(Y_n^m,Y_n^m)_{S_R}\\
&=-R\sum_{n=0}^N\sum_{|m|\le n}
z_n(\kappa R)|v_n^m(R)|^2(Y_n^m,Y_n^m)_{S_1}\\
&=-R\sum_{n=0}^N\sum_{|m|\le n}
z_n(\kappa R)|v_n^m(R)|^2,
\end{align*}
and Lemma~\ref{l:propHankelfunctions} implies
\[
-\Re\,(T_{\kappa,N} v,v)_{S_R}
=R\sum_{n=0}^N\sum_{|m|\le n}
\underbrace{(-\Re\,z_n(\kappa R))}_{\ge 1}|v_n^m(R)|^2
\ge R\sum_{n=0}^N\sum_{|m|\le n} |v_n^m(R)|^2
\ge 0.
\]
\end{proof}
In both cases we obtain the same G{\aa}rding's inequality as in the original
(untruncated) problem Lemma~\ref{l:contGarding}.

The next result is the variational version of the truncation error estimate.
It closely follows the lines of the proof of \cite[Thm.~3.3]{Hsiao:11},
where an estimate of $\|(T_\kappa - T_{\kappa,N})v\|_{s-1/2,2,S_R}$, $s\in\R$,
was proved in the case $d=2$.
\begin{lemma}\label{l:bilintruncerr}
For given $w,v\in H^{1/2}(S_R)$
it holds that
\[
\left|\left((T_\kappa - T_{\kappa,N}) w,v\right)_{S_R}\right| \le c(N,w,v)\|w\|_{1/2,2,S_R}\|v\|_{1/2,2,S_R},
\]
where $c(N,w,v)\ge 0$ and $\lim_{N\to\infty}c(N,w,v) = 0$.
\end{lemma}
\begin{proof}
We start with the two-dimensional situation.
So let
\begin{equation}\label{eq:wvseries2}
\begin{aligned}
w(\vecx) &= w(R\hat\vecx)
= \sum_{|n|\in\N_0} w_n(R)Y_n(\hat\vecx),
\\
v(\vecx) &= v(R\hat\vecx)
= \sum_{|k|\in\N_0} v_k(R)Y_k(\hat\vecx),
\quad
\vecx\in S_R,
\end{aligned}
\end{equation}
be series representations of $w|_{S_R},v|_{S_R}$ with the Fourier coefficients
\begin{align*}
w_n(R) &= (w(R\cdot),Y_n)_{S_1}
= \int_{S_1}w(R\hat\vecx)\overline{Y_n}(\hat\vecx)ds(\hat\vecx),
\\
v_k(R) &= (v(R\cdot),Y_k)_{S_1}
= \int_{S_1}v(R\hat\vecx)\overline{Y_k}(\hat\vecx)ds(\hat\vecx).
\end{align*}
The norm on the Sobolev space $H^s(S_R)$, $s\ge0$,
can be defined as follows \cite[Ch.~1, Rem.~7.6]{Lions:72a}:
\begin{equation}\label{def:SobFourierNorm2}
\|v\|_{s,2,S_R}^2:=
R\sum_{n\in\Z} (1+n^2)^s |v_n(R)|^2.
\end{equation}
Then, by \eqref{def:2dDtNtrunc},
the orthonormality of the circular harmonics \cite[Prop.~3.2.1]{Zeidler:95a}
and \eqref{def:SobFourierNorm2},
\begin{align*}
\left|\left((T_\kappa - T_{\kappa,N}) w,v\right)_{S_R}\right|
&= \frac{1}{R}\left|\sum_{|n|,|k|> N}
\left(Z_n(\kappa R)w_n(R)Y_n(R^{-1}\cdot),v_k(R)Y_k(R^{-1}\cdot)\right)_{S_R}\right|
\\
&= \left|\sum_{|n|,|k|> N}
Z_n(\kappa R)\left(w_n(R)Y_n,v_k(R)Y_k\right)_{S_1}\right|
\\
&= \left|\sum_{|n|> N}
Z_n(\kappa R)w_n(R)\overline{v_n}(R)\right|\\
&= \left|\sum_{|n|> N}
\frac{Z_n(\kappa R)}{(1+n^2)^{1/2}}(1+n^2)^{1/4}w_n(R)(1+n^2)^{1/4}\overline{v_n}(R)\right|\\
&\le \max_{|n|> N}
\left|\frac{Z_n(\kappa R)}{(1+n^2)^{1/2}}\right|
\sum_{|n|> N}\left|(1+n^2)^{1/4}w_n(R)(1+n^2)^{1/4}\overline{v_n}(R)\right|\\
&\le \max_{|n|> N}
\left|\frac{Z_n(\kappa R)}{(1+n^2)^{1/2}}\right|
\left(\sum_{|n|> N}(1+n^2)^{1/2}\left|w_n(R)\right|^2\right)^{1/2}\\
&\quad \times
\left(\sum_{|n|> N}(1+n^2)^{1/2}\left|v_n(R)\right|^2\right)^{1/2}\\
&\le \frac{1}{R}\max_{|n|> N}
\left|\frac{Z_n(\kappa R)}{(1+n^2)^{1/2}}\right|\tilde{c}(N,w,v)\|w\|_{1/2,2,S_R}\|v\|_{1/2,2,S_R},
\end{align*}
where
\[
\tilde{c}(N,w,v)^2
:=\frac{\sum_{|n|> N} (1+n^2)^{1/2}|w_n(R)|^2}%
{\sum_{|n|\in\N_0} (1+n^2)^{1/2}|w_n(R)|^2}
\,\frac{\sum_{|n|> N} (1+n^2)^{1/2}|v_n(R)|^2}%
{\sum_{|n|\in\N_0} (1+n^2)^{1/2}|v_n(R)|^2}.
\]
The coefficient $\tilde{c}(N,w,v)$ tends to zero for $N\to\infty$
thanks to \eqref{def:SobFourierNorm2}, \eqref{def:SobFourierNorm3}..
Corollary~\ref{cor:propHankelfunctions} implies the estimate
\[
\frac{1}{1+n^2} |Z_n(\kappa R)|^2 \le \max\{|Z_0(\kappa R)|^2,1 + |\kappa R|^2\},
\quad
|n|\in\N_0,
\]
hence we can set
\[
c(N,w,v):=\frac{\tilde{c}(N,w,v)}{R}\max\{|Z_0(\kappa R)|,(1 + |\kappa R|^2)^{1/2}\}.
\]
The investigation of the case $d=3$ runs similarly. So let
\begin{equation}\label{eq:wvseries3}
\begin{aligned}
w(\vecx) &= w(R\hat\vecx)
= \sum_{n\in\N_0}\sum_{|m|\le n} w_n^m(R)Y_n^m(\hat\vecx),
\\
v(\vecx) &= v(R\hat\vecx)
= \sum_{k\in\N_0}\sum_{|l|\le k} v_k^l(R)Y_k^l(\hat\vecx),
\quad
\vecx\in S_R,
\end{aligned}
\end{equation}
be series representations of $w|_{S_R},v|_{S_R}$ with the Fourier coefficients
\begin{align*}
w_n^m(R) &= (w(R\cdot),Y_n^m)_{S_1}
= \int_{S_1}w(R\hat\vecx)\overline{Y_n^m}(\hat\vecx)ds(\hat\vecx),
\\
v_k^l(R) &= (v(R\cdot),Y_k^l)_{S_1}
= \int_{S_1}v(R\hat\vecx)\overline{Y_k^l}(\hat\vecx)ds(\hat\vecx).
\end{align*}
The norm on the Sobolev space $H^s(S_R)$, $s\ge0$,
can be defined as follows \cite[Ch.~1, Rem.~7.6]{Lions:72a}:
\begin{equation}\label{def:SobFourierNorm3}
\|v\|_{s,2,S_R}^2:=
R^2\sum_{n\in\N_0}\sum_{|m|\le n} (1+n^2)^s |v_n^m(R)|^2.
\end{equation}
Then, by \eqref{def:3dDtNtrunc},
the orthonormality of the spherical harmonics \cite[Thm.~2.8]{Colton:19}
and \eqref{def:SobFourierNorm3},
\begin{align*}
\left|\left((T_\kappa - T_{\kappa,N}) w,v\right)_{S_R}\right|
&= \frac{1}{R}\left|\sum_{n,k> N}\sum_{|m|\le n,|l|\le k}
\left(z_n(\kappa R)w_n^m(R)Y_n^m(R^{-1}\cdot),v_k^l(R)Y_k^l(R^{-1}\cdot)\right)_{S_R}\right|\\
&= R\left|\sum_{n,k> N}\sum_{|m|\le n,|l|\le k}
z_n(\kappa R)\left(w_n^m(R)Y_n^m,v_k^l(R)Y_k^l)\right)_{S_1}\right|\\
&= R\left|\sum_{n> N}\sum_{|m|\le n}
z_n(\kappa R) w_n^m(R)\overline{v_n^m}(R)\right|\\
&= R\left|\sum_{n> N}\sum_{|m|\le n}
\frac{z_n(\kappa R)}{(1+n^2)^{1/2}}(1+n^2)^{1/4}w_n^m(R)(1+n^2)^{1/4}\overline{v_n^m}(R)\right|\\
&\le R\max_{n> N}
\left|\frac{z_n(\kappa R)}{(1+n^2)^{1/2}}\right|
\sum_{n> N}\sum_{|m|\le n}\left|(1+n^2)^{1/4}w_n^m(R)(1+n^2)^{1/4}\overline{v_n^m}(R)\right|\\
&\le R\max_{n> N}
\left|\frac{z_n(\kappa R)}{(1+n^2)^{1/2}}\right|
\left(\sum_{n> N}\sum_{|m|\le n}(1+n^2)^{1/2}\left|w_n^m(R)\right|^2\right)^{1/2}\\
&\quad \times
\left(\sum_{n> N}\sum_{|m|\le n}(1+n^2)^{1/2}\left|v_n^m(R)\right|^2\right)^{1/2}\\
&\le \frac{1}{R}\max_{n> N}
\left|\frac{z_n(\kappa R)}{(1+n^2)^{1/2}}\right|\tilde{c}(N,w,v)\|w\|_{1/2,2,S_R}\|v\|_{1/2,2,S_R},
\end{align*}
where
\[
\tilde{c}(N,w,v)^2
:=\frac{\sum_{n> N}\sum_{|m|\le n}(1+n^2)^{1/2}\left|w_n^m(R)\right|^2}%
{\sum_{|n|\in\N_0}\sum_{|m|\le n}(1+n^2)^{1/2}\left|w_n^m(R)\right|^2}
\,\frac{\sum_{n> N}\sum_{|m|\le n}(1+n^2)^{1/2}\left|v_n^m(R)\right|^2}%
{\sum_{|n|\in\N_0}\sum_{|m|\le n}(1+n^2)^{1/2}\left|v_n^m(R)\right|^2}.
\]
Thanks to Corollary~\ref{cor:propHankelfunctions} we can define
\[
c(N,w,v) := \frac{\tilde{c}(N,w,v)}{R}\left(2 + |\kappa R|^2\right)^{1/2}.
\]
\end{proof}
\begin{lemma}\label{l:Tkuniformcontinuity}
For $s\in[0,1/2)$ and $w\in H^{1-s}(B_R\setminus\overline\Om)$,
$v\in H^{1+s}(B_R\setminus\overline\Om)$ it holds that
\[
\left|(T_{\kappa,N} w,v)_{S_R}\right|
\le C_\mathrm{bl}\|w\|_{1-s,2,B_R\setminus\overline\Om}\|v\|_{1+s,2,B_R\setminus\overline\Om},
\]
where the constant $C_\mathrm{bl}\ge 0$ does not depend on $N$.
\end{lemma}
\begin{proof}
We start with the two-dimensional situation
as in the proof of Lemma~\ref{l:bilintruncerr}.
If $w,v$ have the representations \eqref{eq:wvseries2},
then, by \eqref{def:2dDtNtrunc},
the orthonormality of the circular harmonics \cite[Prop.~3.2.1]{Zeidler:95a}
and \eqref{def:SobFourierNorm2},
\begin{align*}
\left|(T_{\kappa,N} w,v)_{S_R}\right|
&= \frac{1}{R}\left|\sum_{|n|,|k|\le N}
\left(Z_n(\kappa R)w_n(R)Y_n(R^{-1}\cdot),v_k(R)Y_k(R^{-1}\cdot)\right)_{S_R}\right|
\\
&= \left|\sum_{|n|,|k|\le N}
Z_n(\kappa R)\left(w_n(R)Y_n,v_k(R)Y_k\right)_{S_1}\right|
\\
&= \left|\sum_{|n|\le N}
Z_n(\kappa R)w_n(R)\overline{v_n}(R)\right|\\
&= \left|\sum_{|n|\le N}
\frac{Z_n(\kappa R)}{(1+n^2)^{1/2}}(1+n^2)^{(1/2-s)/2}w_n(R)(1+n^2)^{(1/2+s)/2}\overline{v_n}(R)\right|\\
&\le \max_{|n|\le N}
\left|\frac{Z_n(\kappa R)}{(1+n^2)^{1/2}}\right|
\sum_{|n|\le N}\left|(1+n^2)^{(1/2-s)/2}w_n(R)(1+n^2)^{(1/2+s)/2}\overline{v_n}(R)\right|\\
&\le \max_{|n|\le N}
\left|\frac{Z_n(\kappa R)}{(1+n^2)^{1/2}}\right|
\left(\sum_{|n|\le N}(1+n^2)^{1/2-s}\left|w_n(R)\right|^2\right)^{1/2}\\
&\quad \times
\left(\sum_{|n|\le N}(1+n^2)^{1/2+s}\left|v_n(R)\right|^2\right)^{1/2}\\
&\le \frac{1}{R}\max_{|n|\le N}
\left|\frac{Z_n(\kappa R)}{(1+n^2)^{1/2}}\right|\|w\|_{1/2-s,2,S_R}\|v\|_{1/2+s,2,S_R}.
\end{align*}
Corollary~\ref{cor:propHankelfunctions} implies the estimate
\[
\frac{1}{1+n^2} |Z_n(\kappa R)|^2 \le \max\{|Z_0(\kappa R)|^2,1 + |\kappa R|^2\},
\quad
|n|\in\N_0,
\]
hence
\begin{equation}\label{eq:TkuniformcontinuitySR2}
\left|(T_{\kappa,N} w,v)_{S_R}\right|
\le \frac{1}{R}\max\{|Z_0(\kappa R)|,(1 + |\kappa R|^2)^{1/2}\}\|w\|_{1/2-s,2,S_R}\|v\|_{1/2+s,2,S_R}.
\end{equation}
By the trace theorem \cite[Thm.~3.38]{McLean:00}, we finally arrive at
\[
\left|(T_{\kappa,N} w,v)_{S_R}\right|
\le \frac{C_{tr}^2}{R}\max\{|Z_0(\kappa R)|,(1 + |\kappa R|^2)^{1/2}\}
\|w\|_{1-s,2,B_R\setminus\overline\Om}\|v\|_{1+s,2,B_R\setminus\overline\Om}.
\]
The investigation of the case $d=3$ runs similarly.
So let $w,v$ have the representations \eqref{eq:wvseries3},
then, by \eqref{def:3dDtNtrunc},
the orthonormality of the spherical harmonics \cite[Thm.~2.8]{Colton:19}
and \eqref{def:SobFourierNorm3},
\begin{align*}
\left|(T_{\kappa,N} w,v)_{S_R}\right|
&= \frac{1}{R}\left|\sum_{n,k=0}^N\sum_{|m|\le n,|l|\le k}
\left(z_n(\kappa R)w_n^m(R)Y_n^m(R^{-1}\cdot),v_k^l(R)Y_k^l(R^{-1}\cdot)\right)_{S_R}\right|\\
&= R\left|\sum_{n,k=0}^N\sum_{|m|\le n,|l|\le k}
z_n(\kappa R)\left(w_n^m(R)Y_n^m,v_k^l(R)Y_k^l)\right)_{S_1}\right|\\
&= R\left|\sum_{n=0}^N\sum_{|m|\le n}
z_n(\kappa R) w_n^m(R)\overline{v_n^m}(R)\right|\\
&= R\left|\sum_{n=0}^N\sum_{|m|\le n}
\frac{z_n(\kappa R)}{(1+n^2)^{1/2}}(1+n^2)^{(1/2-s)/2}w_n^m(R)(1+n^2)^{(1/2+s)/2}\overline{v_n^m}(R)\right|\\
&\le R\max_{n\in\N_0}
\left|\frac{z_n(\kappa R)}{(1+n^2)^{1/2}}\right|
\sum_{n=0}^N\sum_{|m|\le n}\left|(1+n^2)^{(1/2-s)/2}w_n^m(R)(1+n^2)^{(1/2+s)/2}\overline{v_n^m}(R)\right|\\
&\le R\max_{n\in\N_0}
\left|\frac{z_n(\kappa R)}{(1+n^2)^{1/2}}\right|
\left(\sum_{n=0}^N\sum_{|m|\le n}(1+n^2)^{1/2-s}\left|w_n^m(R)\right|^2\right)^{1/2}\\
&\quad \times
\left(\sum_{n=0}^N\sum_{|m|\le n}(1+n^2)^{1/2+s}\left|v_n^m(R)\right|^2\right)^{1/2}\\
&\le \frac{1}{R}\max_{n\in\N_0}
\left|\frac{z_n(\kappa R)}{(1+n^2)^{1/2}}\right|\|w\|_{1/2-s,2,S_R}\|v\|_{1/2+s,2,S_R}.
\end{align*}

Corollary~\ref{cor:propHankelfunctions} yields
\begin{equation}\label{eq:TkuniformcontinuitySR3}
\left|(T_{\kappa,N} w,v)_{S_R}\right|
\le \frac{1}{R}\left(2 + |\kappa R|^2\right)^{1/2}\|w\|_{1/2-s,2,S_R}\|v\|_{1/2+s,2,S_R}.
\end{equation}
By the trace theorem \cite[Thm.~3.38]{McLean:00}, we finally arrive at
\[
\left|(T_{\kappa,N} w,v)_{S_R}\right|
\le \frac{C_{tr}^2}{R}\left(2 + |\kappa R|^2\right)^{1/2}
\|w\|_{1-s,2,B_R\setminus\overline\Om}\|v\|_{1+s,2,B_R\setminus\overline\Om}.
\]
\end{proof}

\begin{theorem}\label{th:lintruncprobsolvable}
Under the assumptions of Lemma~\ref{l:contGarding},
given an antilinear continuous functional $\ell:\; V\to\C$, there exists a constant
$N^\ast > 0$ such that for $N \ge N^\ast$ the problem

Find $u_N\in V$ such that
\begin{equation}\label{eq:lintruncproblem}
a_N(u_N,v) = \ell(v)
\quad\text{for all } v\in V
\end{equation}
is uniquely solvable.
\end{theorem}
\begin{proof}
First we show that the problem \eqref{eq:lintruncproblem} has at most one solution.
We start as in the proof of \cite[Thm.~4.5]{Hsiao:11} and argue by contradiction,
i.e.\ we suppose the following:
\begin{equation}\label{eq:contrass}
\begin{gathered}
\forall N^\ast\in\N \quad \exists N = N(N^\ast) \ge N^\ast
\quad\text{and}\quad
u_N = u_{N(N^\ast)}\in V
\quad\text{such that }\\
a_N(u_N ,v) = 0
\quad\text{for all }
v\in V
\quad\text{and }
\|u_N\|_V = 1.
\end{gathered}
\end{equation}

However, the subsequent discussion differs significantly from the proof of \cite[Thm.~4.5]{Hsiao:11}.
We apply an argument the idea of which goes back to Schatz \cite{Schatz:74}.

First we \emph{assume} there exists a solution $u_N\in V$ of \eqref{eq:lintruncproblem}
and derive an a priori estimate of the error $\|u-u_N\|_V$, where $u\in V$ is the
solution of \eqref{eq:intlinproblem}, see Theorem~\ref{th:linprobsolvable}.
Since $a_N$ satisfies a G{\aa}rding's inequality (Lemma~\ref{l:aNprop}(ii)),
we have, making use of \eqref{eq:wavenumbernormequiv},
\[
C_-^2 \|u-u_N\|_V^2 - 2\kappa^2 \|u-u_N\|_{0,2,B_R}^2 \le \Re a_N(u-u_N,u-u_N).
\]
Since
\begin{align*}
a_N(u-u_N,v) &= a_N(u,v) - a_N(u_N,v)\\
&= \underbrace{a(u,v)}_{=\ell(v)} + a_N(u,v) - a(u,v) - \underbrace{a_N(u_N,v)}_{=\ell(v)}\\
&= \left((T_\kappa - T_{\kappa,N}) u,v\right)_{S_R},
\end{align*}
we obtain
\begin{equation}\label{eq:aprioriprim}
C_-^2 \|u-u_N\|_V^2 - 2\kappa^2 \|u-u_N\|_{0,2,B_R}^2
\le \eta_1 \|u-u_N\|_V
\end{equation}
with
\[
\eta_1 := \sup_{v\in V}\frac{\Re \left((T_\kappa - T_{\kappa,N}) u,v\right)_{S_R}}{\|v\|_V}.
\]
Now we consider the following auxiliary adjoint problem (cf.\ \cite[p.~43]{McLean:00}):

Find $w_N\in V$ such that
\[
\overline{a(v,w_N)} = (v,u-u_N)_{B_R}
\quad\text{for all } v\in V.
\]
Since $\mathcal{A}$ is a Fredholm operator (see the proof of Theorem~\ref{th:linprobsolvable}),
the adjoint problem possesses a unique solution $w_N\in V$.
Then
\begin{align*}
\|u-u_N\|_{0,2,B_R}^2 &= \overline{a(u-u_N,w_N)}
= \overline{a(u,w_N)} - \overline{a(u_N,w_N)}\\
&= \underbrace{\overline{a(u,w_N)} - \overline{a_N(u_N,w_N)}}_{=\overline{\ell(w_N)}-\overline{\ell(w_N)}=0}
+ \overline{a_N(u_N,w_N)} - \overline{a(u_N,w_N)}\\
&= \overline{\left((T_\kappa - T_{\kappa,N}) u_N,w_N\right)}_{S_R}.
\end{align*}
In particular, this relation shows that $\left((T_\kappa - T_{\kappa,N}) u_N,w_N\right)_{S_R}$ is real.
With
\[
\eta_2 := \sup_{v\in V}\frac{\left((T_\kappa - T_{\kappa,N}) u_N,v\right)_{S_R}}{\|v\|_V}
\]
we obtain
\[
\|u-u_N\|_{0,2,B_R}^2 \le \eta_2 \|w_N\|_V \le \eta_2 C_-^{-1}C(R,\kappa) \|u-u_N\|_{V^\ast}.
\]
The continuous embedding $V\subset V^\ast$ yields
\[
\|u-u_N\|_{0,2,B_R}^2 \le \eta_2 C_-^{-1}C(R,\kappa) C_\mathrm{emb}\|u-u_N\|_V.
\]
Applying this estimate in \eqref{eq:aprioriprim}, we get
\[
C_-^2 \|u-u_N\|_V^2 -2\kappa^2\eta_2 C_-^{-1}C(R,\kappa) C_\mathrm{emb}\|u-u_N\|_V
\le \eta_1 \|u-u_N\|_V.
\]
Now, if $\|u-u_N\|_V\ne 0$, we finally arrive at
\begin{equation}\label{eq:intermedee}
C_-^2 \|u-u_N\|_V \le \eta_1 + 2\kappa^2\eta_2 C_-^{-1}C(R,\kappa) C_\mathrm{emb}.
\end{equation}
Clearly this inequality is true also for $\|u-u_N\|_V = 0$ so that we can
remove this interim assumption.

Thanks to Lemma~\ref{l:bilintruncerr} we have that
\[
\left|\left((T_\kappa - T_{\kappa,N}) u,v\right)_{S_R}\right| \le c(N,u,v)\|u\|_{1/2,2,S_R}\|v\|_{1/2,2,S_R}
\le c(N,u,c) C_\mathrm{tr}^2 \|u\|_V\|v\|_V,
\]
hence
\begin{equation}\label{eq:eta1est}
\eta_1 \le c_+(N,u) C_\mathrm{tr}^2 \|u\|_V
\quad\text{with}\quad
c_+(N,u):=\sup_{v\in V} c(N,u,v),
\end{equation}
where $\lim_{N\to\infty}c_+(N,u) = 0$.
Note that, as can be seen from the proof of Lemma~\ref{l:bilintruncerr},
the second fractional factor in the representation of $\tilde{c}(N,w,v)$
can be estimated from above by one without losing the limit behaviour for $N\to\infty$.
Consequently, $\eta_1$ can be made arbitrarily small provided $N$ is large enough.

In order to estimate $\eta_2$ we cannot apply Lemma~\ref{l:bilintruncerr} directly
since the second argument in the factor $c(N,u_N,v)$ depends on $N$, too.
Therefore we give a more direct estimate.

Namely, let $v\in V$ have the representation \eqref{eq:wvseries2} or \eqref{eq:wvseries3}, respectively.
Then we define
\[
V_N|_{S_R} := \begin{cases}
\spann_{|n|\le N}\{Y_n(R^{-1}\cdot)\},& d=2,\\
\spann_{n=0\ldots N,|m|\le n}\{Y_n^m(R^{-1}\cdot)\},& d=3,
\end{cases}
\]
and introduce an orthogonal projector
\[
P_N:\;V|_{S_R}\to V_N|_{S_R}:\; v\mapsto P_Nv := \begin{cases}
\sum_{|n|\le N} v_n(R)Y_n(R^{-1}\cdot),
& d=2,\\
\sum_{n=0}^N\sum_{|m|\le n} v_n^m(R)Y_n^m(R^{-1}\cdot),
& d=3.
\end{cases}
\]
Then it holds that $V_N|_{S_R}\subset\ker (T_\kappa P_N - T_{\kappa,N})$.
Indeed, if $d=2$ and $v\in V_N|_{S_R}$, then
$P_Nv = v = \sum_{|n|\le N} v_n(R)Y_n(R^{-1}\cdot)$ and
\[
T_\kappa P_Nv = T_\kappa v
= \frac{1}{R}\sum_{|n|\le N} Z_n(\kappa R)v_n(R)Y_n(R^{-1}\cdot)
= T_{\kappa,N} v.
\]
An analogous argument applies in the case $d=3$.

Now we return to the estimate of $\eta_2$ and write, for $u_N\in V$,
\[
(T_\kappa - T_{\kappa,N}) u_N
=(T_\kappa - T_\kappa P_N) u_N + (T_\kappa P_N - T_{\kappa,N}) u_N
=T_\kappa (\id - P_N) u_N,
\]
where we have used the above property.
The advantage of this approach is that we can apply a wellknown estimate
of the projection error.
The proof of this estimate runs similarly to the proof of Lemma~\ref{l:bilintruncerr}
but only without the coefficients $Z_n$ or $z_n$, respectively:
\begin{align*}
\left|\left((\id - P_N) w,v\right)_{S_R}\right|
&= \left|\sum_{|n|,|k|> N}
\left(w_n(R)Y_n(R^{-1}\cdot),v_k(R)Y_k(R^{-1}\cdot)\right)_{S_R}\right|
\\
&= R\left|\sum_{|n|,|k|> N}
\left(w_n(R)Y_n,v_k(R)Y_k\right)_{S_1}\right|
\\
&= R\left|\sum_{|n|> N}
w_n(R)\overline{v_n}(R)\right|\\
&= R\left|\sum_{|n|> N}
\frac{1}{(1+n^2)^{1/2}}(1+n^2)^{1/4}w_n(R)(1+n^2)^{1/4}\overline{v_n}(R)\right|\\
&\le \max_{|n|> N}
\frac{R}{(1+n^2)^{1/2}}
\sum_{|n|> N}\left|(1+n^2)^{1/4}w_n(R)(1+n^2)^{1/4}\overline{v_n}(R)\right|\\
&\le \frac{R}{(1+N^2)^{1/2}}
\left(\sum_{|n|> N}(1+n^2)^{1/2}\left|w_n(R)\right|^2\right)^{1/2}\\
&\quad \times
\left(\sum_{|n|> N}(1+n^2)^{1/2}\left|v_n(R)\right|^2\right)^{1/2}\\
&\le \frac{1}{(1+N^2)^{1/2}}\|w\|_{1/2,2,S_R}\|v\|_{1/2,2,S_R}.
\end{align*}
The same estimate holds true for $d=3$.
Then we get, by Remark~\ref{rem:sharpDtNbound} (or Lemma~\ref{l:Tkuniformcontinuity}),
\begin{align*}
\left|\left((T_\kappa - T_{\kappa,N}) u_N,v\right)_{S_R}\right|
&=\left|\left(T_\kappa (\id - P_N) u_N,v\right)_{S_R}\right|\\
&\le \frac{C\kappa}{(1+N^2)^{1/2}}\|u_N\|_{1/2,2,S_R}\|v\|_{1/2,2,S_R}\\
&\le \frac{CC_\mathrm{tr}^2\kappa}{(1+N^2)^{1/2}}\|u_N\|_V\|v\|_V,
\end{align*}
thus
\[
\eta_2 \le \frac{CC_\mathrm{tr}^2\kappa}{(1+N^2)^{1/2}} \|u_N\|_V.
\]
Using this estimate and \eqref{eq:eta1est} in \eqref{eq:intermedee}, we obtain
\begin{equation}\label{eq:schatzest}
C_-^2 \|u-u_N\|_V \le c_+(N,u) C_\mathrm{tr}^2 \|u\|_V + 2\kappa^2 C_-^{-1}C(R,\kappa) C_\mathrm{emb}
\frac{CC_\mathrm{tr}^2\kappa}{(1+N^2)^{1/2}} \|u_N\|_V.
\end{equation}
Now we apply this estimate to the solutions $u_N$ of the homogeneous truncated problems
in \eqref{eq:contrass}.
By Theorem~\ref{th:linprobsolvable}, the homogeneous linear interior problem \eqref{eq:intlinproblem}
(i.e.\ $\ell=0$) has the solution $u=0$, and the above estimate implies
\[
C_-^2 \|u_N\|_V \le 2\kappa^2 C_-^{-1}C(R,\kappa) C_\mathrm{emb}
\frac{CC_\mathrm{tr}^2\kappa}{(1+N^2)^{1/2}} \|u_N\|_V,
\]
which is a contradiction to $\|u_N\|_V=1$ for all $N$.
\end{proof}
Although the proof of Theorem~\ref{th:lintruncprobsolvable} allows an analogous conclusion
as in Lemma~\ref{l:infsup} that the truncated sesquilinear form $a_N$ satisfies an inf-sup condition,
such a conclusion is not fully satisfactory since the question remains
whether and how the  inf-sup constant depends on $N$ or not.
However, at least for sufficiently large $N$, a positive answer can given.
\begin{lemma}\label{l:truncinfsup}
Under the assumptions of Lemma~\ref{l:contGarding},
there exists a number $N^\ast\in\N$ such that
\[
\beta_{N^\ast}(R,\kappa) := \inf_{w\in V\setminus\{0\}}\sup_{v\in V\setminus\{0\}}
\frac{|a_N(w,v)|}{\|w\|_{V,\kappa} \|v\|_{V,\kappa}}>0
\]
is independent of $N\ge N^\ast$.
\end{lemma}
In the proof a formula is given that expresses $\beta_{N^\ast}(R,\kappa)$
in terms of $\beta(R,\kappa)$.
\begin{proof}
We return to the proof of Theorem~\ref{th:lintruncprobsolvable} and
mention that the estimate \eqref{eq:schatzest} is valid for
solutions $u,u_N$ of the general linear problems \eqref{eq:intlinproblem}
(or, equally, \eqref{eq:intlinopproblem}) and \eqref{eq:lintruncproblem}, respectively.

By the triangle inequality,
\begin{align*}
\|u_N\|_V &\le \|u\|_V + \|u-u_N\|_V \\
&\le \|u\|_V + c_+(N,u) C_-^{-2}C_\mathrm{tr}^2 \|u\|_V + 2\kappa^2 C_-^{-3}C(R,\kappa) C_\mathrm{emb}
\frac{CC_\mathrm{tr}^2\kappa}{(1+N^2)^{1/2}} \|u_N\|_V.
\end{align*}
If $N^\ast$ is sufficiently large such that
\[
\kappa^2 C_-^{-3}C(R,\kappa) C_\mathrm{emb}
\frac{CC_\mathrm{tr}^2\kappa}{(1+N^2)^{1/2}} \le \frac14
\quad\text{and}\quad
c_+(N,u) C_-^{-2}C_\mathrm{tr}^2 \le 1
\quad\text{for all }N\ge N^\ast,
\]
then, by Lemma~\ref{l:infsup},
\[
\|u_N\|_V \le 4 \|u\|_V \le \frac{4}{C_-} \|u\|_{V,\kappa}
\le  \|\ell\|_{V^\ast}.
\]
That is, the sesquilinear form $a_N$ satisfies an \emph{inf-sup condition}
\[
\beta_{N^\ast}(R,\kappa) := \inf_{w\in V\setminus\{0\}}\sup_{v\in V\setminus\{0\}}
\frac{|a_N(w,v)|}{\|w\|_{V,\kappa} \|v\|_{V,\kappa}} > 0
\]
with $\displaystyle\beta_{N^\ast}(R,\kappa):= \frac{C_- \beta(R,\kappa)}{4C_+}$
independent of $N\ge N^\ast$.
\end{proof}
Analogously to \eqref{eq:defopA} we introduce the truncated linear operator
$\mathcal{A}_N:\; V \to V^\ast$ by
\[
\mathcal{A}_N w(v) := a_N(w,v)
\quad\text{for all } w,v\in V.
\]
By Lemma~\ref{l:aNprop}, $\mathcal{A}_N$ is a bounded operator,
and Lemma~\ref{l:truncinfsup} implies that $\mathcal{A}_N$ has a bounded inverse:
\[
\|w\|_{V,\kappa} \le \beta_{N^\ast}(R,\kappa)^{-1}\|\mathcal{A}_N w\|_\ast
\quad\text{for all } w\in V.
\]
Furthermore, we define a nonlinear operator $\mathcal{F}_N:\; V \to V^\ast$ by
\[
\mathcal{F}_N(w)(v) := \ell^\mathrm{contr}(w) + \ell^\mathrm{src}(w) + \ell_N^\mathrm{inc}
\qquad\text{for all }
w\in V,
\]
where
\[
\langle\ell_N^\mathrm{inc},v\rangle
:=
(\vecr\cdot\nabla u^\mathrm{inc} - T_{\kappa,N} u^\mathrm{inc},v)_{S_R}.
\]
The problem \eqref{eq:intscalproblem4} is then equivalent to the operator equation
\[
\mathcal{A}_N u = \mathcal{F}_N(u)
\quad\text{in } V^\ast,
\]
and further to the fixed-point problem
\begin{equation}\label{eq:scalfpptruncated}
u = \mathcal{A}_N^{-1}\mathcal{F}_N(u)
\quad\text{in } V.
\end{equation}\begin{theorem}\label{th:truncscalee}
Under the assumptions of Lemma~\ref{l:contGarding},
let the functions $c$ and $f$ generate locally Lipschitz continuous
Nemycki operators in $V$ and assume that there exist functions $w_f,w_c\in V$ such that
$f(\cdot,w_f)\in L_{p_f/(p_f-1)}(\Om)$ and $c(\cdot,w_f)\in L_{p_c/(p_c-2)}(\Om)$,
respectively.

Furthermore let
$u^\mathrm{inc}\in H^1_\mathrm{loc}(\Om^+)$ be such that
additionally $\Delta u^\mathrm{inc}\in L_{2,\mathrm{loc}}(\Om^+)$ holds.

If there exist numbers $\varrho>0$ and $L_\mathcal{F}\in (0,\beta_{N^\ast}(R,\kappa))$
(where $N^\ast$ and $\beta_{N^\ast}(R,\kappa)$ are from Lemma~\ref{l:truncinfsup})
such that the following two conditions
\begin{align*}
&\kappa^2\left[\|c(\cdot,w_c)-1\|_{0,\tilde{q}_c,\Om}
+ \|L_c(\cdot,w,w_c)\|_{0,q_c,\Om}(\varrho + \|w_c\|_V)\right]\varrho\nonumber\\
&\quad + \left[\|f(\cdot,w_f)\|_{0,\tilde{q}_f,\Om}
+ \|L_f(\cdot,w,w_f)\|_{0,q_f,\Om}(\varrho + \|w_f\|_V)\right]\\
&\quad
+ C_\mathrm{tr}\|\vecr\cdot\nabla u^\mathrm{inc} - T_{\kappa,N} u^\mathrm{inc}\|_{-1/2,2,S_R}
\le \varrho\beta_{N^\ast}(R,\kappa),\nonumber\\
&\kappa^2\left[\|L_c(\cdot,w,v)\|_{0,q_c,\Om}\varrho
+ \|c(\cdot,w_c)-1\|_{0,\tilde{q}_c,\Om}
+ \|L_c(\cdot,w,w_c)\|_{0,q_c,\Om}(\varrho + \|w_c\|_V)\right]\nonumber\\
&\quad + \|L_f(\cdot,w,v)\|_{0,q_f,\Om} \le L_\mathcal{F}
\end{align*}
are satisfied for all $w,v\in K_\varrho^\mathrm{cl}$,
then the problem \eqref{eq:scalfpptruncated} has a unique solution
$u_N \in K_\varrho^\mathrm{cl}$ for all $N \ge N^\ast$.
\end{theorem}
\begin{proof}
Analogously to the proof of Theorem~\ref{th:scalee}.
\end{proof}

The next result is devoted to an estimate of the truncation error $\|u-u_N\|_V$.
We start from the proof of Theorem~\ref{th:lintruncprobsolvable} but have in mind the
nonlinear problems \eqref{eq:intscalproblem3} and \eqref{eq:intscalproblem4}.
So let $u,u_N\in V$ be the solutions of \eqref{eq:intscalproblem3} and \eqref{eq:intscalproblem4},
respectively.
Since $a_N$ satisfies a G{\aa}rding's inequality by Lemma~\ref{l:aNprop}, we have that
\begin{equation}\label{eq:gaardingappl}
C_-^2 \|u-u_N\|_V^2 - 2\kappa^2 \|u-u_N\|_{0,2,B_R}^2 \le \Re a_N(u-u_N,u-u_N),
\end{equation}
where
$C_- := \min\{1;\kappa\}$.
In order to estimate the right-hand side, we write
\begin{align*}
a_N(u-u_N,v) &= a_N(u,v) - a_N(u_N,v)\\
&= a(u,v) + a_N(u,v) - a(u,v) - a_N(u_N,v)\\
&= a_N(u,v) - a(u,v) + n(u,v) - n_N(u_N,v).
\end{align*}
Now, the first difference term is equal to $\left((T_\kappa - T_{\kappa,N}) u,v\right)_{S_R}$,
and for the second one we have
\begin{align*}
&\quad
n(u,v) - n_N(u_N,v)\\
&= \kappa^2(c(\cdot,u) - 1) u,v)_{B_R}
+ (f(\cdot,u),v)_{B_R}
- (T_\kappa u^\mathrm{inc},v)_{S_R}
+ (\vecr\cdot\nabla u^\mathrm{inc},v)_{S_R}\\
&\quad
-\kappa^2(c(\cdot,u_N) - 1) u_N,v)_{B_R}
- (f(\cdot,u_N),v)_{B_R}
+ (T_{\kappa,N} u^\mathrm{inc},v)_{S_R}
- (\vecr\cdot\nabla u^\mathrm{inc},v)_{S_R}\\
&= \kappa^2(c(\cdot,u) - 1) u,v)_{B_R}
-\kappa^2(c(\cdot,u_N) - 1) u_N,v)_{B_R}
+ (f(\cdot,u),v)_{B_R}
- (f(\cdot,u_N),v)_{B_R}\\
&\quad
- \left((T_\kappa - T_{\kappa,N}) u^\mathrm{inc},v\right)_{S_R}.
\end{align*}
As in the proof of Theorem~\ref{th:scalee} we see that
\begin{align*}
&\quad
|n(u,v) - n_N(u_N,v)|\\
&\le \|\ell^\mathrm{contr}(u) - \ell^\mathrm{contr}(u_N)\|_{V^\ast}\|v\|_V
+ \|\ell^\mathrm{src}(u) - \ell^\mathrm{src}(u_N)\|_{V^\ast}\|v\|_V
+ \left|\left((T_\kappa - T_{\kappa,N}) u^\mathrm{inc},v\right)_{S_R}\right|\\
&\le L_\mathcal{F} \|u-u_N\|_V\|v\|_V + \eta^\mathrm{inc}\|v\|_V
\end{align*}
with
\[
\eta^\mathrm{inc}
:= \sup_{v\in V}\frac{\left|\left((T_\kappa - T_{\kappa,N}) u^\mathrm{inc},v\right)_{S_R}\right|}{\|v\|_V}.
\]
Hence we obtain from \eqref{eq:gaardingappl}
\begin{equation}\label{eq:aprioriprim}
\begin{aligned}
C_-^2 \|u-u_N\|_V^2 - 2\kappa^2 \|u-u_N\|_{0,2,B_R}^2
&\le \eta_1 \|u-u_N\|_V + L_\mathcal{F} \|u-u_N\|_V^2 + \eta^\mathrm{inc}\|u-u_N\|_V\\
&= \left(\eta_1 + \eta^\mathrm{inc} + L_\mathcal{F} \|u-u_N\|_V\right)\|u-u_N\|_V
\end{aligned}
\end{equation}
with
\[
\eta_1 := \sup_{v\in V}\frac{\Re \left((T_\kappa - T_{\kappa,N}) u,v\right)_{S_R}}{\|v\|_V}.
\]
Now we consider the following auxiliary adjoint problem (cf.\ \cite[p.~43]{McLean:00}):

Find $w_N\in V$ such that
\[
\overline{a(v,w_N)} = (v,u-u_N)_{B_R}
\quad\text{for all } v\in V.
\] 
Since $\mathcal{A}$ is a Fredholm operator (see the proof of Theorem~\ref{th:linprobsolvable}),
the adjoint problem possesses a unique solution $w_N\in V$.
Then
\begin{align}
\|u-u_N\|_{0,2,B_R}^2 &= \overline{a(u-u_N,w_N)}
= \overline{a(u,w_N)} - \overline{a(u_N,w_N)}
\nonumber\\
&= \overline{a(u,w_N)} - \overline{a_N(u_N,w_N)}
+ \overline{a_N(u_N,w_N)} - \overline{a(u_N,w_N)}
\label{eq:sqL2err}\\
&= \overline{n(u,w_N)} - \overline{n_N(u_N,w_N)}
+ \overline{a_N(u_N,w_N)} - \overline{a(u_N,w_N)}.
\nonumber
\end{align}
The first difference term can be estimated as above, that is,
\begin{align*}
\left|\overline{n(u,w_N)} - \overline{n_N(u_N,w_N)}\right|
\le L_\mathcal{F} \|u-u_N\|_V\|w_N\|_V + \eta^\mathrm{inc}\|w_N\|_V.
\end{align*}
Setting
\[
\eta_2 := \sup_{v\in V}\frac{\left|\left((T_\kappa - T_{\kappa,N}) u_N,v\right)_{S_R}\right|}{\|v\|_V}\,,
\]
we obtain from \eqref{eq:sqL2err}
\begin{align*}
\|u-u_N\|_{0,2,B_R}^2
&\le \left(L_\mathcal{F} \|u-u_N\|_V
+ \eta^\mathrm{inc} +\eta_2\right) \|w_N\|_V\\
&\le \left(L_\mathcal{F} \|u-u_N\|_V
+ \eta^\mathrm{inc} +\eta_2\right) C_-^{-1}C(R,\kappa) \|u-u_N\|_{V^\ast},
\end{align*}
and the continuous embedding $V\subset V^\ast$ yields
\[
\|u-u_N\|_{0,2,B_R}^2 \le \left(L_\mathcal{F} \|u-u_N\|_V
+ \eta^\mathrm{inc} +\eta_2\right) C_-^{-1}C(R,\kappa) C_\mathrm{emb}\|u-u_N\|_V.
\]
Applying this estimate in \eqref{eq:aprioriprim}, we get
\begin{gather*}
C_-^2 \|u-u_N\|_V^2 -2\kappa^2\left(L_\mathcal{F} \|u-u_N\|_V
+ \eta^\mathrm{inc} +\eta_2\right) C_-^{-1}C(R,\kappa) C_\mathrm{emb}\|u-u_N\|_V\\
\le \left(\eta_1 + \eta^\mathrm{inc} + L_\mathcal{F} \|u-u_N\|_V\right) \|u-u_N\|_V.
\end{gather*}
Now, if $\|u-u_N\|_V\ne 0$, we finally arrive at
\begin{equation}\label{eq:intermedee}
\begin{gathered}
C_-^2 \|u-u_N\|_V \\
\le \eta_1 + \eta^\mathrm{inc} + L_\mathcal{F} \|u-u_N\|_V
+ 2\kappa^2\left(L_\mathcal{F} \|u-u_N\|_V
+ \eta^\mathrm{inc} +\eta_2\right) C_-^{-1}C(R,\kappa) C_\mathrm{emb}.
\end{gathered}
\end{equation}
Clearly this inequality is true also for $\|u-u_N\|_V = 0$ so that we can
remove this interim assumption.

Thanks to Lemma~\ref{l:bilintruncerr} we have that
\[
\left|\left((T_\kappa - T_{\kappa,N}) u,v\right)_{S_R}\right|
\le c_+(N,u)\|u\|_{1/2,2,S_R}\|v\|_{1/2,2,S_R}
\]
with $\lim_{N\to\infty}c_+(N,u) = 0$,
hence
\begin{equation}\label{eq:eta1est}
\eta_1 \le c_+(N,u) C_\mathrm{tr}^2 \|u\|_V.
\end{equation}
Consequently, $\eta_1$ can be made arbitrarily small provided $N$ is large enough.
Analogously, $\eta^\mathrm{inc}$ can be made arbitrarily small for sufficiently large $N$.

For $\eta_2$, we have the following estimate (see the proof of Theorem~\ref{th:lintruncprobsolvable}):
\[
\eta_2 \le \frac{CC_\mathrm{tr}^2\kappa}{(1+N^2)^{1/2}} \|u_N\|_V.
\]
Using this estimate, the inequality \eqref{eq:eta1est}, and the analogous estimate for $\eta^\mathrm{inc}$
in \eqref{eq:intermedee}, we obtain
\begin{align*}
C_-^2 \|u-u_N\|_V
&\le c_+(N,u) C_\mathrm{tr}^2 \|u\|_V
+ c_+(N,u^\mathrm{inc}) C_\mathrm{tr}^2 \|u^\mathrm{inc}\|_V
+ L_\mathcal{F} \|u-u_N\|_V\\
&\quad
+ 2\kappa^2 C_-^{-1}C(R,\kappa) C_\mathrm{emb}
\Big(L_\mathcal{F} \|u-u_N\|_V
+ c_+(N,u^\mathrm{inc}) C_\mathrm{tr}^2 \|u^\mathrm{inc}\|_V\\
&\qquad
+ \frac{CC_\mathrm{tr}^2\kappa}{(1+N^2)^{1/2}} \|u_N\|_V\Big).
\end{align*}
Now, if $L_\mathcal{F}$ also satisfies
\[
\tilde{\varrho}\big(1+2\kappa^2 C_-^{-1}C(R,\kappa) C_\mathrm{emb}\big) L_\mathcal{F}\le \frac{C_-^2}{4},
\]
where we have used the (pessimistic) estimate
\[
\|u-u_N\|_V \le \|u\|_V + \|u_N\|_V
\le 2\tilde{\varrho}
\]
with $\tilde{\varrho}$ being the maximum value of the radii $\varrho$ from
the nonlinear existence and uniqueness theorems for $u$ and $u_N$, respectively
(cf.~Theorems~\ref{th:scalee}, \ref{th:lintruncprobsolvable}),
we conclude
\begin{equation}\label{eq:truncerrest1}
\begin{aligned}
C_-^2 \|u-u_N\|_V
&\le 2c_+(N,u) C_\mathrm{tr}^2 \|u\|_V
+ 2c_+(N,u^\mathrm{inc}) C_\mathrm{tr}^2 \|u^\mathrm{inc}\|_V\\
&\quad
+ 4\kappa^2 C_-^{-1}C(R,\kappa) C_\mathrm{emb}
\Big(c_+(N,u^\mathrm{inc}) C_\mathrm{tr}^2 \|u^\mathrm{inc}\|_V
+ \frac{CC_\mathrm{tr}^2\kappa}{(1+N^2)^{1/2}} \|u_N\|_V\Big).
\end{aligned}
\end{equation}
We have proved the following result.
\begin{theorem}\label{th:truncscalerrest}
Let the assumptions of the above Theorem~\ref{th:truncscalee}
with respect to $R$, $\kappa$, $c$, and $f$ be satisfied.
Then, if the Lipschitz constant $L_\mathcal{F}$ is sufficiently small,
i.e.\ satisfies
\[
L_\mathcal{F}\le\min\left\{\beta(R,\kappa),\beta_{N^\ast}(R,\kappa),
\frac{C_-^2}{4\tilde{\varrho}\big(1+2\kappa^2 C_-^{-1}C(R,\kappa) C_\mathrm{emb}\big)}
\right\},
\]
there exists a constant $c>0$ independent of $N\ge N^\ast$
(the structure of the constant can be seen from \eqref{eq:truncerrest1})
such that the following estimate holds:
\[
c\|u-u_N\|_V
\le c_+(N,u) \|u\|_V
+ c_+(N,u^\mathrm{inc}) \|u^\mathrm{inc}\|_V\\
+ \frac{1}{(1+N^2)^{1/2}} \|u_N\|_V.
\]
\end{theorem}

\section{Conclusion}
A mathematical model together with an investigation of existence
and uniqueness of its solution for radiation and propagation effects
on compactly supported nonlinearities is presented.
The full-space problem is reduced to an equivalent truncated local problem,
whereby in particular the dependence of the solution on the truncation parameter
(with regard to stability and error of the truncated solution) is studied.
The results form the basis for the use of numerical methods, e.g., FEM,
for the approximate solution of the original problem with controllable accuracy.

\end{document}